\documentclass[11pt]{article}

\usepackage{amsmath,amsfonts,amsthm,amssymb,color}
\usepackage{fancybox}
\usepackage{tikz}
\usepackage[algosection]{algorithm2e}
\usepackage{hyperref}
\usepackage{cleveref} % must come after hyperref
\usepackage{thm-restate}
\crefformat{equation}{(#2#1#3)}

\usepackage{caption}
\usepackage{bbm}
\usepackage{tablefootnote}

\SetAlCapSkip{.5em}

\usepackage[margin=1.0in]{geometry}

\newtheorem{theorem}{Theorem}[section]

\newtheorem{lemma}[theorem]{Lemma}

\theoremstyle{definition}
\newtheorem{definition}[theorem]{Definition}
\newtheorem{remark}[theorem]{Remark}

\newenvironment{fminipage}%
{\begin{Sbox}\begin{minipage}}%
		{\end{minipage}\end{Sbox}\fbox{\TheSbox}}

\newenvironment{algbox}[0]{\vskip 0.2in
	\noindent
	\begin{fminipage}{6.3in}
	}{
\end{fminipage}
\vskip 0.2in
}

\newcommand{\myparagraph}[1]{\smallskip\noindent {\bf #1.}}
\newcommand{\classP}{\ensuremath{\mathsf{P}}}
\newcommand{\Pcomplete}{$\mathsf{P}$-complete}

\newcommand{\NCone}{\ensuremath{\mathsf{NC}^{1}}}

\newcommand{\MPC}{\ensuremath{\mathsf{MPC}}}

\def\defeq{\stackrel{\mathrm{def}}{=}}

\def\ContractionSampling{\textsf{ContractionSampling\ }}

\newcommand\PPi{\boldsymbol{\Pi}}

\renewcommand\deg{\boldsymbol{\mathit{d}}}

\newcommand\xx{\boldsymbol{\mathit{x}}}

\newcommand\DD{\boldsymbol{\mathit{D}}}

\newcommand\LL{\boldsymbol{\mathit{L}}}

\newcommand\Otil{\widetilde{O}}

\newcommand\xxhat{\boldsymbol{\widehat{\mathit{x}}}}

\newcommand\edge{e}

\newcommand\xor{\bigoplus}

\title{Parallel Batch-Dynamic Graphs: Algorithms and Lower Bounds}

\author{
	David Durfee \\
	LinkedIn \\
	\texttt{ddurfee@linkedin.com}
	\and
	Laxman Dhulipala\\
	CMU\\%Carnegie Mellon University\\
	\texttt{ldhulipa@cs.cmu.edu}
	\and
	Janardhan Kulkarni \\
	Microsoft Research \\
	\texttt{jakul@microsoft.com}
	\and
	Richard Peng \\
	Georgia Tech \\
	\texttt{richard.peng@gmail.com}
	\and
	Saurabh Sawlani \\
	Georgia Tech \\
	\texttt{sawlani@gatech.edu}
	\and
	Xiaorui Sun \\
	University of Illinois at Chicago \\
	\texttt{xiaorui@uic.edu}
}

\date{}

\begin{document}

\maketitle

%!TEX root = Main.tex

\begin{abstract}
In this paper we study the problem of dynamically maintaining graph
properties under batches of edge insertions and deletions in the
massively parallel model of computation. In this setting, the graph is
stored on a number of machines, each having space strongly sublinear
with respect to the number of vertices, that is, $n^\epsilon$ for
some constant $0 < \epsilon < 1$. Our goal is to handle batches of
updates and queries where the data for each batch fits onto one
machine in constant rounds of parallel computation, as well as to
reduce the total communication between the machines. This objective
corresponds to the gradual buildup of databases over time, while the
goal of obtaining constant rounds of communication for problems in the
static setting has been elusive for problems as simple as undirected
graph connectivity.

We give an algorithm for dynamic graph connectivity in this setting with constant
communication rounds and communication cost almost linear in terms of the batch size.
Our techniques combine a new graph contraction technique,
an independent random sample extractor from correlated samples,
as well as distributed data structures supporting
parallel updates and queries in batches.

We also illustrate the power of dynamic algorithms in the \MPC{} model 
by showing that the batched version of the adaptive connectivity problem is 
$\mathsf{P}$-complete in the centralized setting, but sub-linear sized
batches can be handled in a constant number of rounds. Due to the
wide applicability of our approaches, we believe it represents a
practically-motivated workaround to the current difficulties in
designing more efficient massively parallel static graph algorithms.

%We also illustrate the power of dynamic algorithms in the \MPC{}
%model by giving a constant-round algorithm handling sub-linear sized
%batches for the adaptive connectivity problem, a problem which we show
%to be $\mathsf{P}$-complete in the centralized setting.  Due to the
%wide applicability of our approaches, we believe it represents a
%practically-motivated workaround to the current difficulties in
%designing more efficient massively parallel static graph algorithms.

\end{abstract}

\thispagestyle{empty}
\newpage
\pagenumbering{arabic}

%!TEX root = Main.tex

\section{Introduction}\label{sec:intro}

Parallel computation frameworks and storage systems, such as
MapReduce, Hadoop and Spark, have been proven to be highly effective
methods for representing and analyzing the massive datasets that
appear in the world today.
Due to the importance of this new class of systems, models of
parallel computation capturing the power of such systems have been
increasingly studied in recent years, with the Massively Parallel Computation (\MPC{})
model~\cite{KSV10} now serving as the canonical model.
In recent years the \MPC{} model has seen the development of algorithms
for fundamental problems, including clustering~\cite{EIM11, BMVKV12,
BBLM14, YV17},
connectivity problems~\cite{RMCS13, KLMRV14,ASSWZ18,ASW18,
andoni2019log},
optimization~\cite{MKSK13, EN15, PENW16},
dynamic programming~\cite{IMS16, bateni2018massively}, to name several
as well as many other fundamental graph and optimization problems~\cite{BKV12,andoni2014parallel,
KMV15,AG18, AK17, AK17b, ASSWZ18, brandt2018matching,
czumaj2018round, lkacki2018connected, onak2018round, ABBMS17, assadi2019sublinear,
assadi2019distributed,
behnezhad2019massively, behnezhad2019exponentially, gamlath2018weighted, ghaffari2019conditional,
ghaffari2019sparsifying, hajiaghayi2019massively}.
Perhaps the main goal in these algorithms has been solving the
problems in a constant number of communication rounds while minimizing
the total communication in a round.
Obtaining low round-complexity is well motivated due to the high cost
of a communication round in practice, which is often between minutes
and hours~\cite{KSV10}.
Furthermore, since communication between processors tends to be much
more costly than local computation, ensuring low communication
per-round is also an important criteria for evaluating algorithms in
the \MPC{} model~\cite{sarma2013upper, beame2013communication}.

Perhaps surprisingly, many natural problems such as dynamic
programming~\cite{IMS16} and submodular maximization~\cite{PENW16} can
in fact be solved or approximated in a constant number of
communication rounds in \MPC{} model.
However, despite considerable effort, we are still far from obtaining
constant-round algorithms for many natural problems in the \MPC{} setting
where the space-per-machine is restricted to be sublinear in the
number of vertices in the graph (this setting is arguably the most
reasonable modeling choice, since real-world graphs can have trillions
of vertices). For example, no constant round algorithms are known for
a problem as simple as connectivity in an undirected graph, where the
current best bound is $O( \log{n})$ rounds in general~\cite{KSV10,RMCS13, KLMRV14, ASSWZ18, lkacki2018connected, ASW18}.
Other examples include a $O(\sqrt{\log n})$ round algorithm for
approximate graph
matching~\cite{onak2018round,ghaffari2019sparsifying}, and
a $O(\sqrt{\log \log n})$-round algorithm for $(\Delta+1)$ vertex
coloring~\cite{chang2019coloring}.
Even distinguishing between a single cycle of size $n$ and two cycles
of size $n/2$ has been conjectured to require $\Omega(\log{n})$
rounds~\cite{KSV10, RMCS13, KLMRV14, RVW16, YV17, ghaffari2019conditional, Sungjin19}.
Based on this conjecture,
recent studies have shown that several other graph related problems,
such as maximum matching, vertex cover,  maximum independent set
and single-linkage clustering
cannot be solved in a constant number of rounds~\cite{YV17, ghaffari2019conditional}.
%
%
%, and
%recent studies in this topic have focused on parameterizing in terms
%of diameter~\cite{ASSWZ18} and spectrum~\cite{ASW18}.

On the other hand, most large-scale databases are not formed by a single
atomic snapshot, but form rather gradually through an accretion of
updates. Real world examples of this include the construction of
social networks~\cite{LiuZY10}, the accumulation of log
files~\cite{HingaveI15}, or even the gradual change of the Internet
itself~\cite{DeanG08,KangTF09,MalewiczAMBDHLC10}. In each of these
examples, the database is gradually formed over a period of months, if
not years, of updates, each of which is significantly smaller than the
whole database. It is often the case that the updates are grouped
together, and are periodically processed by the database as a batch.
Furthermore, it is not uncommon to periodically re-index the data
structure to handle a large number of queries between sets of updates.

In this paper, motivated by the gradual change in real-world datasets
through batches of updates, we consider the problem of maintaining
graph properties in dynamically changing graphs in the \MPC{} model.
Our objective is to maintain the graph property for \emph{batches} of
updates, while achieving a \emph{constant number of rounds} of
computation in addition to also minimizing the total \emph{communication}
between machines in a given round.

Specifically, we initiate the study of \emph{parallel batch-dynamic
graph problems} in \MPC{}, in which an update contains a number of mixed
edge insertions and deletions. We believe that batch-dynamic
algorithms in \MPC{} capture the aforementioned real world examples of
gradually changing databases, and provide an efficient distributed
solution when the size of the update is large compared to single update
dynamic algorithms. We note that a similar model for dynamic graph
problems in \MPC{} was recently studied by Italiano et
al.~\cite{italiano2019dynamic}. However, they focus on the scenario
where every update only contains a single edge insertion or deletion.
Parallel batch-dynamic algorithms were also recently studied in the
shared-memory setting by Tseng et al.~\cite{tseng2018batch} for the
forest-connectivity problem and Acar et al.~\cite{acar2019parallel}
for dynamic graph connectivity. However, the depth of these algorithms
is at least $\Omega(\log n)$, and it is not immediately clear whether
these results can be extended to low (constant) round-complexity
batch-dynamic algorithms in the \MPC{} setting.

We also study the power of dynamic algorithms in the \MPC{} setting by
considering a natural ``semi-online" version of the connectivity
problem which we call \emph{adaptive connectivity}. We show that the
adaptive connectivity problem is $\mathsf{P}$-complete, and therefore
in some sense inherently sequential, at least in the centralized
setting. In contrast to this lower bound in the centralized setting,
we show that in the \MPC{} model there is a batch-dynamic algorithm that
can process adaptive batches with size proportional to the space
per-machine in a constant number of rounds. Note that such an
algorithm in the centralized setting (even one that ran in slightly
sublinear depth per batch) would imply an algorithm for the Circuit
Value Problem with polynomial speedup, thus solving a longstanding
open problem in the parallel complexity landscape.

%!TEX root = Main.tex

\subsection{Our Results}

Since graph connectivity proves to be an effective representative for the aforementioned difficulty of graph problems in the \MPC{} model,
the focus of this paper is studying graph connectivity and adaptive graph connectivity in the batch-dynamic \MPC{} model.

\subsubsection*{Graph Connectivity}

%\paragraph{Connectivity, 2-Edge-Connectivity, 3-Edge-Connectivity:}
The dynamic connectivity problem is to determine if a given pair of
vertices belongs to same connected component in the graph as the graph
undergoes (batches of) edge insertions and deletions. The dynamic
connectivity algorithm developed in this paper is based on a
hierarchical partitioning scheme that requires a
more intricate incorporation of sketching based data structures for
the sequential setting.
Not only does our scheme allow us to achieve a constant number of
rounds, but it also allows us to achieve a total communication bound
that is linear with respect to the batch size with only an additional
$n^{o(1)}$ factor.

\begin{restatable}[]{theorem}{Main}
%\begin{theorem}
\label{thm:Main}
In the \MPC{} model with memory per machine $s = \Otil(n^\epsilon)$
we can maintain a dynamic undirected graph on $m$ edges which,
for constants $\delta, \alpha$, and integer $k$ such that $k \cdot
n^{\alpha + \delta} \cdot \mathrm{polylog}(n) \leq s$, can handle the
following operations with high probability:
\begin{enumerate}
  \item A batch of up to $k$ edge insertions/deletions, using $O(1/(\delta\alpha))$ rounds.
  \item
  Query up to $k$ pairs of vertices for 1-edge-connectivity, using $O(1/\alpha)$ rounds.
\end{enumerate}
Furthermore, the total communication for handling a batch of $k$
operations is $\Otil(k n^{\alpha + \delta})$, and the total space used across all machines is $\Otil(m)$.
%\end{theorem}
\end{restatable}

\subsubsection*{Adaptive Connectivity and Lower-Bounds in the Batch-Dynamic \MPC{} Model}

In the \emph{adaptive connectivity problem}, we are given a sequence
of query/update pairs. The problem is to process each query/update pair in
order, where each query determines whether or not a given pair of
vertices belongs to the same connected component of the graph, and
applies the corresponding dynamic update to the graph if the query
succeeds. We obtain the following corollary by applying our batch-dynamic connectivity algorithm, Theorem~\ref{thm:Main}.

\begin{restatable}[]{corollary}{AdaptiveUpperBound}
In the \MPC{} model with memory per machine $s = \Otil(n^\epsilon)$
we can maintain a dynamic undirected graph on $m$ edges
which for constants $\delta, \alpha$, and integer $k$ such that $k \cdot n^{\alpha + \delta} \cdot \mathrm{polylog}(n) \leq s$
can handle the following operation with high probability:
\begin{enumerate}
  \item An adaptive batch of up to $k$ (query, edge
    insertions/deletions) pairs, using $O(1/(\delta\alpha))$ rounds.
\end{enumerate}
Furthermore, the total communication for handling a batch of $k$
operations is $\Otil(k n^{\alpha + \delta})$, and the total space used across
all machines is $\Otil(m)$.
\end{restatable}

We also provide a lower-bound for the adaptive connectivity problem in
the centralized setting, showing that the problem is
$\mathsf{P}$-complete under $\mathsf{NC}^{1}$ reduction.
$\mathsf{P}$-completeness is a standard notion of parallel hardness~\cite{kruskal1990complexity, Greenlaw1995, blelloch1996parallel}.
As a consequence of our reduction, we show that the adaptive
connectivity algorithm does not admit a parallel algorithm in the
centralized setting with polynomial speedup, unless the
(Topologically-Ordered) Circuit Value Problem admits a parallel
algorithm with polynomial speedup, which is a long-standing open
problem in parallel complexity literature.

\begin{restatable}[]{theorem}{AdaptiveLowerBound}
\label{thm:1connlb}
%\begin{theorem}\label{thm:1connlb}
The adaptive connectivity problem is $\mathsf{P}$-complete under
$\mathsf{NC}^{1}$ reductions.
%\end{theorem}
\end{restatable}

By observing that our reduction, and the $\mathsf{NC}^{1}$ reductions
proving the hardness for the Circuit Value Problem can be done in
$O(1)$ rounds of \MPC{}, we have the following corollary in the \MPC{}
setting.
\begin{restatable}[]{corollary}{MPCLowerBound}
\label{thm:mpclowerbound}
In the \MPC{} model with memory per machine $s = \Otil(n^\epsilon)$ for
some constant $\epsilon$, if adaptive connectivity on a sequence of
size $O(n)$ can be solved in $O(k)$ rounds, then every problem in
$\mathsf{P}$ can be solved in $O(k)$ rounds.
\end{restatable}

\subsection{Batch-Dynamic \MPC{} Model}\label{sec:dynamic_mpc_model}
In this section, we first introduce the massively parallel computation (\MPC{}) model,
followed by the batch-dynamic \MPC{} model which is the main focus of
this paper.

\paragraph{Massively Parallel Computation (\MPC{}) Model.}
The Massively Parallel Computation (\MPC{}) model
is a widely accepted theoretical model for parallel computation~\cite{KSV10}.
Here, the input graph $G$ has $n$ vertices and at most $m$ edges
at any given instant.
We are given $p$ processors/machines, each with local memory for storage $s = \widetilde{\Theta}(m/p)$.\footnote{Throughout this paper, $\widetilde{\Theta}$ and $\Otil$ hide polylogarithmic terms in the size of the input.}
Note that we usually assume $m^{1-\delta} \geq p \geq m^{\delta}$, for some $\delta>0$.
This is because the model is relevant only when the number of machines and
the local memory per machine are significantly smaller than the size of the input.

The computation in the \MPC{} model proceeds via rounds.
Initially, the input data is distributed across the processors arbitrarily.
During each round, each processor runs a polynomial-time algorithm
on the data which it contains locally.
Between rounds, each machine receives at most $\mu$ amount of data from other machines.
The total data received by all machines between any two rounds is termed as the communication cost.
Note that no computation can occur between rounds, and equivalently,
no communication can occur during a round.

The aim for our algorithms in this model is twofold.
Firstly and most importantly, we want to minimize the number of rounds required for our algorithm,
since this cost is the major bottleneck of massively parallel algorithms in practice.
Ideally, we would want this number to be as low as $O(1)$.
Secondly, we want to decrease the maximum communication cost over all rounds,
since the costs of communication between processors in practice are massive
in comparison to local computation.

\paragraph{Batch-Dynamic \MPC{} Model.}
%\todo{should this be in its own section?}
At a high-level, our model works as follows.
Similar to recent works by Acar et al.~\cite{acar2019parallel} and Tseng et al.~\cite{tseng2018batch}, we assume that the graph undergoes
batches of insertions and deletions, and in the initial round of each computation,
an update or query batch is distributed to an arbitrary machine.  %that we can distribute to the desired
%processors.
%\saurabh{The second part of the previous sentence is unclear - what are we trying to say?}\xiaorui{fixed}
The underlying computational model used is the \MPC{} model,
and assume that space per machine is strongly sublinear with respect
to the number of vertices of the graph, that is, $O(n^\alpha)$ for
some constant $0 < \alpha < 1$.

More formally, we assume there are two kinds of operations in a batch:
\begin{enumerate}
\item Update: A set of edge insertions/deletions of size up to $k$.
\item Query: A set of graph property queries of size up to $k$.
\end{enumerate}
%Furthermore, we will consider batches of queries being made, and the number of updates or queries in a batch will never exceed the size of any processor.
For every batch of updates, the algorithm needs to properly maintain the graph according to the edge insertions/deletions such that
the algorithm can accurately answer a batch of queries at any instant.
We believe that considering batches of updates and queries most closely relates to practice where often multiple updates occur in the examined network before another query is made.
Furthermore, in the \MPC{} model there is a distinction between a batch of updates and a single update, unlike the standard model, because it is possible for the batch update to be made in parallel,
and handling batch updates or queries is as efficient as handling a single update or query, especially in terms of the number of communication rounds.
%, which can still keep the number of rounds constant and also minimize the total communication.

We use  two criteria to measure the efficiency of parallel dynamic algorithms: the number of  communication rounds and the total communication between different machines.
Note that massively parallel algorithms for static problems are often most concerned with communication rounds.
In contrast, we also optimize the total communication in the dynamic setting, since the total communication becomes a bottleneck for practice when overall data size is very huge, especially when the update is much smaller than the total information of the graph.
Ideally, we want to handle batches of updates and queries in constant communication rounds and sublinear  total communication with respect to the number of vertices in the graph.

The key algorithmic difference between the dynamic model we introduce here and the \MPC{} model is that we can decide how to partition the input into processors as updates occur to the graph.

Dynamic problems in the \MPC{} model were  studied in the very
recent paper by Italiano et al.~\cite{italiano2019dynamic}.
Their result only explicitly considers the single update case.
%Although we believe it can be  extended to handle batches as well.
In the batch-dynamic scenario, the result of \cite{italiano2019dynamic} generalizes but has higher dependencies on batch
sizes in both number of rounds and total communication.%but is less efficient in terms of number of communication rounds or the total communication.
Our incorporation of graph  sketching, fast contraction, and batch search trees are all critical for obtaining our optimized dependencies on batch sizes.

%!TEX root = Main.tex

\subsection{Our Techniques}\label{sec:intro_techniques}

In this section we give in-depth discussion of the primary techniques
used to achieve the results presented in the previous section.

\subsubsection*{Connectivity}

Without loss of generality, we assume that the batch of updates is either only edge insertions or only edge deletions.
For a mixed update batch with both insertions and deletions, we can simply handle the edge deletions first, and then the edge insertions.
In case the same edge is being inserted and deleted, we simply eliminate both operations.

Similar to previous results on dynamic connectivity~\cite{Frederickson85, GalilI92, HenzingerK99,HolmDT01,ahn2012analyzing, KapronKM13,
	GibbKKT15:arxiv,NanongkaiS17,Wulffnilsen17,NanongkaiSW17},
we maintain a maximal spanning forest.
This forest encodes the connectivity information in the graph,
and more importantly,
%undergoes at most one update per update
undergoes few changes per update to the graph.
Specifically:
\begin{enumerate}
	\item An \emph{insert} can cause at most two trees in $F$ to be joined to form
	a single tree.
	\item A \emph{delete} may split a tree into two, but if there
	exists another edge between these two resulting trees, they should
	then be connected together to ensure that the forest is maximal.
\end{enumerate}

Our dynamic trees data structure adapts the recently
developed parallel batch-dynamic data structure for maintaining a maximal spanning forest in the shared-memory
setting by Tseng et al.~\cite{tseng2018batch} to the \MPC{} model.
Specifically, \cite{tseng2018batch} give a
parallel batch-dynamic algorithm that runs in $O(\log n)$ depth w.h.p.
to insert $k$ new edges to the spanning forest,
to remove $k$ existing edges in the spanning forest,
or to query the IDs of the spanning tree containing the given $k$ vertices.
We show that the data structure can be modified to achieve $O(1/\alpha)$
round-complexity and $O(k\cdot n^{\alpha})$ communication for any small constant $\alpha$ satisfying $k \cdot n^\alpha \cdot \mathrm{polylog}(n) \leq s$ in the \MPC{} setting.
In addition, if we associate with each vertex a key of length $\ell_{key}$,
then we can query and update a batch of $k$ key values in $O(1/\alpha)$
round-complexity and $O(k \cdot \ell_{key}\cdot)$ communication.

With a parallel batch-dynamic data structure to maintain a maximal spanning forest,
a batch of edge insertions or edge queries for the dynamic
connectivity problem can be handled  in $O(1/\alpha)$
round-complexity and $O(k\cdot n^\alpha)$ communication for any constant $\alpha$.
Our strategy for insertions and queries is similar to the dynamic
connectivity algorithm of Italiano et al.~\cite{italiano2019dynamic}:
A set of edge queries can be handled by querying the IDs of the spanning tree of all the vertices involved.
Two vertices are in the same connected component if and only if their IDs are equal.
To process a batch of edge insertions, we maintain the maximal spanning forest by first identifying
the set of edges in the given batch that join different spanning trees without creating cycles using ID queries, and then inserting these edges to the spanning forest, by linking their respective trees.

Handling a set of edge deletions, however, is more complex.
This is because if some spanning forest edges are removed,
then we need to find replacement edges which are in the graph, but previously not in the spanning forest,
that can be added to the spanning forest without creating cycles.
To facilitate this,
we incorporate developments in sketching
based sequential data structures for dynamic connectivity~\cite{ahn2012analyzing, KapronKM13}.

To construct a sketch of parameter $0 < p < 1$ for a graph, we first independently sample every edge of the graph with probability $p$, and then set the sketch for each vertex to be the XOR of the IDs for all the sampled edges which are incident to the vertex.
A sketch has the property that for any subset of vertices, the XOR of the sketches of these vertices equals to the XOR of the IDs for all the sampled edges leaving the vertex subset.
In particular, if there is only a single sampled edge leaving the vertex subset, then the XOR of the sketches of these vertices equals to the ID of the edge leaving the vertex subset.

The high level idea of \cite{ahn2012analyzing, KapronKM13} is to use sketches for each current connected component to sample previous non-tree edges going out of the connected component using sketches with different parameters, and use these edges to merge
connected components that are separated after deleting some tree edges.
We visualize this process by representing each connected component as a vertex in a multigraph,
and finding a replacement non-tree edge between the two components as the process of merging these
two vertices.
At first glance, it seems like we can translate this approach
to the \MPC{} model by storing all the sketches for each connected component
in a single machine.
However, directly translating such a data structure leads to either $\text{polylog}(n)$ communication rounds or $\Omega(m)$ total communication per update batch.
To see this, let us look at some intuitive ideas to adapt this data structure to the \MPC{} model,
and provide some insight into why they have certain limitations:
\begin{enumerate}
\item \textbf{Sketch on the original graph}:
For this case, once we use the sketch to sample an edge going out of a given connected component,
we only know the ID of the two vertices of the edge, but not the two connected components the edge connects.
Obtaining the information about which connected components the endpoints belong to requires communication,
because a single machine cannot store the connected component ID of each vertex in the graph.
Hence, to contract all the connected components using sampled edges for each connected component, we need one round of communication.
Since we may need to reconnect as many as $k$ connected components ($k$ is the number of deletions, i.e., the batch size),
this approach could possibly require $\log k = \Theta(\log n)$ communication rounds.
\item \textbf{Sketch on the contracted graph where every connected component is contracted to a single vertex}:
To do this, each edge needs to know which connected components its endpoints belong to.
If we split a connected component into several new connected components after deleting some tree edges,
the edges whose vertices previously belong to same connected component may now belong to different connected components.
To let each edge know which connected components its endpoints belong to,
we need to broadcast the mapping between vertices and connected components to all the related edges.
Hence, the total communication can be as large as $\Omega(m)$.
To further illustrate this difficulty via an example,
consider the scenario that the current maximal spanning forest is a path of $n$ vertices, and
a batch of $k$ edge deletions break the path into $k+1$ short paths.
In this case, almost all the vertices change their connected component IDs.
In order to
find edges previously not in the maximal spanning forest to link these $k+1$ path,
every edge needs to know if the two vertices of the edge belong to same connected component or not,
and to do this, the update of connected component ID for vertices of every edge requires $\Omega(m)$ communication.
%\xiaorui{One STOC reviewer have following concern "One could potentially have a set of machines to maintain the component id of every node in the graph so that every machine i maintains all the id for machines with ids between $[i n^\delta, (i+1) n^\delta]$. Or am I missing something?". Is it ok now?}
\end{enumerate}

The high level idea of our solution is to speed up the ``contraction" process such that constant iterations suffice to shrink all the connected components into a single vertex.
To do this, sampling $\Otil(1)$ edges leaving each  connected component in each iterations (as previous work) is not enough,
because of the existence of low conductance graph.
Hence,
we need to sample a much larger number of edges leaving each connected component.
Following this intuition, we prove a \textbf{fast contraction lemma} which shows that picking $n^\alpha$
edges out of each component finds all connecting non-tree edges between components
within $O(1/\alpha)$ iterations.

However, a complication that arises with the aforementioned fast contraction lemma is that
it requires the edges leaving a component to be \emph{independently} sampled.
But the edges sampled by a single sketch are correlated.
This correlation comes from the fact that
a sketch outputs an edge leaving a connected component if and only if there is only one sampled edge leaving that connected component.
To address this issue, we construct an \textbf{{independent sample extractor}} to identify enough edges that are eventually sampled independently based on the sketches
and show that these edges are enough to simulate the independent sampling process required by the fast contraction lemma. % contraction process.

We discuss these two ideas in depth below. 
In the rest of this section, 
we assume without loss of generality that every current connected component is contracted into a single vertex, since the sampled edges are canceled under the XOR operation for sketches. 

%To get around of this difficulty,

%We give a fast contraction lemma which allows to use multiple sketches to reduce the number of edges with a rate polynomial of the number of sketches,
%and an extractor to obtain independent edge samples from correlated edges given by sketches.
%Directly translating such a data structure does give $\text{polylog}(n)$ communication
%per update, but at the cost of a prohibitive $\text{polylog}(n)$ rounds per update.

%\vspace{.3cm}
\paragraph{Fast Contraction Lemma.}
We first define a random process for edge sampling (which we term \ContractionSampling)
in Definition~\ref{def:contractionsampling}.
The underlying motivation for such a definition is that the edges
obtained from the sketch are not independently sampled.
So, we tweak the sampling process via an independent sample extractor,
which can then produce edges which obey the random process \ContractionSampling.
Before discussing this independent sample extractor,
we will first outline how edges sampled using \ContractionSampling
suffice for fast contraction.
%We instead incorporate the sketching based routine into our hierarchical partitioning.
%Specifically,
%we give adaptations of the tree data structures and the underlying
%dynamic graph algorithms in Sections~\ref{sec:1conn}
%and~\ref{sec:MPCDynamicTree}, and make use of $n^\delta$ sketches to speed up the contraction process by defining
%a so-called ContractionSampling process(formally defined in \ref{def:contractionsampling}) and making use of the following lemma. %Lemma~\ref{lem:contraction_main}.

\begin{restatable}[\ContractionSampling process]{definition}{ContractionSamplingDef}
\label{def:contractionsampling}
%\begin{definition}
The random process \ContractionSampling  for a multigraph $G = (V, E)$ and an integer $k$ is defined as follows:
	each vertex $v$ independently draws $t_v$ samples $S_{v, 1}, S_{v, 2}, \dots S_{v, t_v}$
	for some integer $t_v \geq k$ such that
		\begin{enumerate}
		\item %For any $i \in t_v$,
		the outcome of each $S_{v, i}$ can be an either an edge incident to $v$ or $\bot$;
		\item
		for every edge $e$ incident to vertex $v$,
		\[\sum_{i = 1}^{t_v} \Pr[S_{v, i} = e] \geq \Omega\left(\frac{k \log^2 n}{\deg_G(v)}\right).\]
		\end{enumerate}
		%Let $E'$ be the set of all edges sampled in this manner by at least one vertex. }
%\end{definition}
\end{restatable}

We show that in each connected component, if we contract edges sampled by the \ContractionSampling process,
the number of edges remaining reduces by a polynomial factor with high probability by taking $k = \text{poly}(n)$.

\begin{restatable}[]{lemma}{Contraction}

%\begin{lemma}
\label{lem:contraction_main}
	Consider the following contraction scheme starting with a multigraph $G(V,E)$
	on $n$ vertices and $m < poly(n)$ (multi) edges:
	For a fixed integer $k$,
	\begin{enumerate}
		\item
		let $E'$ be a set of edges sampled by the \ContractionSampling process;
		\item contract vertices belonging to same connected component of graph $G' = (V, E')$ into a new graph $G^\star = (V^\star, E^\star)$ as follows:
		each vertex of $V^\star$ represents a connected component in the sampled graph $G' = (V, E')$,
		and there is an edge between two vertices $x, y \in V^\star$ iff
		there is an edge in $G$ between the components corresponding to $x$ and $y$,
		with edge multiplicity equal to the sum of multiplicity of edges in $G$ between the components
		corresponding to $x$ and $y$.
	\end{enumerate}
	Then the resultant graph has at most $\widetilde O(m k^{-1/3})$ (multi) edges with high probability.

%\end{lemma}
\end{restatable}

%{
%\color{red}
%In particular, we show that
%for a multigraph on $n$ vertices and $m < poly(n)$ (multi) edges,
Based on Lemma~\ref{lem:contraction_main}, if we iteratively apply the \ContractionSampling process with $k=n^{\alpha}$ and shrink connected components using sampled edges into a single vertex,
then every connected component of the multigraph becomes a singleton vertex in $O( 1 / \alpha)$ rounds with high probability.

%
%\begin{restatable}[]{lemma}{Contraction}
%	\label{lem:Contraction}
%	Consider an iterative contraction scheme starting with a multigraph
%	on $n$ vertices and $m < poly(n)$ (multi) edges that performs
%	at each round:
%	\begin{enumerate}
%		\item for each vertex $u$, we pick
%		$n^{\delta}$ independent neighbors of it uniformly at random.
%		\item then we contract all edges picked.
%	\end{enumerate}
%	Then the graph becomes singleton vertices (of which there may be
%	multiple as the initial graph may be disconnected)
%	in $O( 1 / \delta)$ rounds with high probability.
%\end{restatable}

Lemma~\ref{lem:contraction_main}
can be shown using a straightforward argument for simple graphs.
However, in the case of multigraphs (our graphs are multigraphs because there
can be more than one edge between two components), this argument
is not as easy. It is possible that for a connected component $C_1$,
a large number of edges leaving $C_1$ will go to another connected component $C_2$.
Hence, in one round, the sampled $n^\delta$ edges leaving $C_1$ may all go to $C_2$.
From this perspective, we cannot use a simple degree-based counting argument to show that every connected component merges with at least $n^\delta$ other connected components
if it connected to at least $n^\delta$ other connected components.
%}

To deal with parallel edges, and to prove that the contraction
occurs in constant, rather than $O(\log n)$ rounds,
we make use of a more combinatorial analysis.
Before giving some intuition about this proof, we define some useful
terminology.

\begin{restatable}[Conductance]{definition}{Conductance}
	\label{def:conductance}
	Given a graph $G(V,E)$ and a subset of vertices $S \subseteq V$,
the conductance of $S$ w.r.t. $G$ is defined as
\[
\phi_G(S) \defeq \min_{S' \subseteq S} \dfrac{\left| E(S', S \setminus S') \right|}{\min \left\{ \sum_{u \in S'} \deg_G(u), \sum_{u \in S \setminus S'} \deg_G(u) \right\}}.
\]
\end{restatable}

The conductance of a graph is a measure of how ``well-knit" a graph is.
Such graphs are of consequence to us because the more well-knit the graph is,
the faster it contracts into a singleton vertex.
We use the expander decomposition lemma from \cite{SpielmanT11},
which says that any connected multigraph $G$ can be partitioned into
such subgraphs.

\begin{restatable}[\cite{SpielmanT11}, Section 7.1.]{lemma}{Expander}
		\label{lem:Expander}
		Given a parameter $k > 0$,  
		any graph $G$ with $n$ vertices and $m$ edges can be partitioned into
		groups of vertices $S_1, S_2, \ldots$ such that
		\begin{itemize}
			\item the conductance of each $S_i$ is at least $1/k$;
			\item the number of edges between the $S_i$'s is at most $O(m \log{n}/k)$.
		\end{itemize}	
\end{restatable}

For each such ``well-knit" subgraph $H$ to collapse in one round of sampling,
the sampled edges in $H$ must form a spanning subgraph of $H$.
One way to achieve this is to generate a spectral sparsifier of $H$ \cite{SpielmanS08} -
which can be obtained by sampling each edge with a probability at least $O(\log n)$ times
its \emph{effective resistance}.
The effective resistance of an edge is the amount of current that would pass through
it when unit voltage difference is applied across its end points, which is a measure
of how important it is to the subgraph being well-knit.

As the last piece of the puzzle, we show that the edges sampled by the \ContractionSampling process
do satisfy the required sampling constraint to produce a spectral sparsifier
of $H$.
Since each such subgraph collapses, Lemma~\ref{lem:Expander} also tells us
that only a small fraction of edges are leftover in $G$,
as claimed in Lemma~\ref{lem:contraction_main}.

It is important to note that although we introduce sophisticated tools such as expander partitioning
and spectral sparsifiers, these tools are only used in the proof
and not in the actual algorithm to find replacement edges.

%\xiaorui{Add more technical details here.}
%These tools are closely related to parallel
%and distributed binary search trees~\cite{EllenFRV10,DrachslerVY14},
%and parallel contraction routines~\cite{MillerR89,AcarAW17} respectively.
%However, the more stringent requirement of constant (instead of $O(\log{n})$
%as in NC) rounds necessitates both a more combinatorial analysis, and the
%use of more sophisticated tools such as expander partitioning
%to help prove the convergence of the contraction algorithm.

%\input{2connectivity_tech}
%\input{23connectivity_tech}
%Theorem~\ref{thm:Main} is based on the following a contraction lemma  which might be of independent interest.

\paragraph{From Sketches to Independent Samples.}

On a high level, our idea to achieve fast contraction
is based on using $O(k \cdot \text{polylog}(n))$ independent sketches.
However, we cannot directly claim that these sketches simulate a \ContractionSampling procedure, as required by the fast contraction lemma (Lemma~\ref{lem:contraction_main}).
This is because \ContractionSampling requires the edges being sampled independently.
Instead, each sketch as given by \cite{ahn2012analyzing, KapronKM13} gives a set of edges are constructed as follows:
\begin{enumerate}
\item Pick each edge independently with probability $p$, where $p$ is the parameter of the sketch.
\item For each vertex which has exact one sampled edge incident to it, output the sampled incident edge. 
%For each subset of the partition, if there is one edge incident to it, output it.
\end{enumerate}
The second step means the samples picked out of two vertices are correlated.
Given a vertex $v$, 
let $E_v$ be the random variable for the edge picked in Step 2 of above sketch construction process.
Consider an example with two adjacent vertices $v_1$ and $v_2$. 
If the outcome of $E_{v_1}$ is the edge $v_1v_2$, then 
the outcome of $E_{v_2}$ cannot be an edge other than $v_1v_2$.
Hence two random variables $E_{v_1}$ and $E_{v_2}$ are correlated.

This issue is a direct side-effect of the faster contraction procedure.
Previous uses of sketching only needs to find one edge leaving per component, which suffices for $O(\log n)$ rounds.
However, our goal is to terminate in a constant number of rounds.
This means we need to claim much larger connected components among the sampled edges.
For this purpose, we need independence because most results on independence between edges require some correlation between the edges picked.

Instead, we show that each sketch still generates a large number of independent edge samples.
That is, while all the samples generated by a copy of the sketch are dependent on each other, a sufficiently large subset of it is in fact, independent.
Furthermore, observe that contractions can only make more progress when more edges are considered.
So it suffices to show that this particular subset we choose makes enough progress.
Formally, we prove the following lemma.
\begin{restatable}{lemma}{IndependenceExtractor}
%\begin{lemma}
\label{lem:independence}
Given an integer $k$ and a multigraph $G$ of $n$ vertices,
$ O(k \log^3 n)$ independent sketches simulates a \ContractionSampling process.
Furthermore, for every edge sampled by the \ContractionSampling process,
there exists a sketch and a vertex such that the value of the sketch on the vertex is exactly the ID of that edge.
%\end{lemma}
\end{restatable}

Our starting observation is that for a bipartite graph, 
sketching process gives independent edge samples for vertices from the same side:
For a bipartite graph $(A, B)$, %any two vertices $v_1, v_2 \in A$, $E_{v_1}$ and $E_{v_2}$ are independent, since the outcome of $E_{v_1}$ depends on samples for edges incident to $v_1$,  the outcome of $E_{v_2}$ depends on samples for edges incident to $v_2$, and these two sets of edges are disjoint. 
the process of sampling edges, and picking all edges incident to degree one vertices of $A$
satisfies the property that all the edges picked are independent. 

To extend this observation to general graph, we consider a bipartition of the graph, $(A, B)$, and view the random sampling of edges from the sketch as a two-step process:
\begin{enumerate}
\item First, we sample all edges within each bipartition $(A, A)$ and $(B, B)$.
\item Then we sample the $(A, B)$ edges independently.
\end{enumerate}
After first step, we remove vertices from $A$ that have some sampled  edges incident to.
The second step gives a set of edges, from which we keep ones incident to some degree one vertices from $A$.
Based on the observation of bipartite graph, the edges kept in the second step are independent (condition on the outcome of the first step).

To bound the probability of picking an edge crossing the bipartition, 
we will first lower bound the probability that the incident vertex from $A$ remains after the first step, and 
then check that the second step on the bipartite graph is equivalent to an independent process on the involved edges. 
The overall lower bound on the probability of an edge picked then follows from combining the probability of an edge being picked in one of these processes with the probability that the corresponding vertices remain after the first step and the initial pruning of vertices.
With this probability estimation, we show that $O(k \cdot \text{polylog} n)$ independent sketches are enough to boost the probability of picking the edge to the required lower bound by the \ContractionSampling process.

At the end, we show that $O(\log n)$ random bipartition of the graph is enough to make sure that every edge appears in at least one of the bipartition, and then Lemma~\ref{lem:independence} follows.

\subsubsection*{Adaptive Connectivity and Lower-Bounds in the Batch-Dynamic \MPC{} Model}

The adaptive connectivity problem is the ``semi-online" version of the
connectivity problem where the entire adaptive batch of operations is
given to the algorithm in advance, but the algorithm must apply the
query/update pairs in the batch \emph{in order}, that is each pair
on the graph defined by applying the prefix of updates before it. We
note that the problem is closely related to \emph{offline} dynamic
problems, for example for offline dynamic minimum spanning tree and
connectivity~\cite{Eppstein1994}. The main difference is that in the
offline problem the updates (edge insertions/deletions) are not
adaptive, and are therefore not conditionally run based on the
queries.  We also note here that every problem that admits a static
$\mathsf{NC}$ algorithm also admits an $\mathsf{NC}$ algorithm for the
offline variant of the problem. The idea is to run, in
parallel for each query, the static algorithm on the input graph
unioned with the prefix of the updates occuring before the query.
Assuming the static algorithm is in $\mathsf{NC}$, this gives a
$\mathsf{NC}$ offline algorithm (note that
obtaining work-efficient parallel offline algorithms for problems like
minimum spanning tree and connectivity is an interesting problem that
we are not aware of any results for).

Compared to this positive result in the setting without adaptivity,
the situation is very different once the updates are allowed to
adaptively depend on the results of the previous query, since the
simple black-box reduction given for the offline setting above is no
longer possible.
In particular, we show the following lower bound for the adaptive connectivity
problem which holds in the centralized setting: the adaptive
connectivity problem is
$\mathsf{P}$-complete, that is unless $\mathsf{P} = \mathsf{NC}$,
there is no $\mathsf{NC}$ algorithm for the problem. The adaptive
connectivity problem is clearly in $\mathsf{P}$ since we can just run
a sequential dynamic connectivity algorithm to solve it.
To prove the hardness result, we give a low-depth reduction from the
Circuit Value Problem (CVP), one of the canonical
$\mathsf{P}$-complete problems. The idea is to take the gates in the
circuit in some
topological-order (note that the version of CVP where the gates are
topologically ordered is also $\mathsf{P}$-complete), and transform
the evaluation of the circuit into the execution of an adaptive
sequence of connectivity queries. We give an $\mathsf{NC}^{1}$
reduction which evaluates a circuit using adaptive connectivity
queries as follows.  The reduction maintains that all gates that
evaluate to $\mathsf{true}$ are contained in a single connected
component connected to some root vertex, $r$. Then, to determine
whether the next gate in the topological order, $g = g_a \wedge g_b$,
evaluates to $\mathsf{true}$ the reduction runs a connectivity query
testing whether the vertices corresponding to $g_a$ and $g_b$ are
connected in the current graph, and adds an edge $(g, r)$, thereby
including it in the connected component of $\mathsf{true}$ gates if
the query is true. Similarly, we reduce evaluating $g = g_a \vee g_b$
gates to two queries, which check whether $g_a$ ($g_b$) is reachable
and add an edge from $(g, r)$ in either case if so. A $g = \lnot g_a$
gate is handled almost similarly, except that the query checks whether
$g_a$ is disconnected from $s$. Given the topological ordering of the
circuit, generating the sequence of adaptive queries can be done in
$O(\log n)$ depth and therefore the reduction works in
$\mathsf{NC}^{1}$.

In contrast, in the \MPC{} setting, we show that we can achieve $O(1)$
rounds for adaptive batches with size proportional to the space per
machine. Our algorithm for adaptive connectivity follows naturally
from our batch-dynamic connectivity algorithm based on the following
idea: we assume that every edge deletion in the batch actually occurs,
and compute a set of replacement edges in $G$ for the (speculatively)
deleted edges. Computing the replacement edges can be done in the same
round-complexity and communication cost as a static batch of deletions
using Theorem~\ref{thm:Main}. Since the number of replacement edges is
at most $O(k) = O(s)$, all of the replacements can be sent to a single
machine, which then simulates the sequential adaptive algorithm on the
graph induced by vertices affected by the batch in a single round. We
note that the upper-bound in \MPC{} does not contradict the
$\mathsf{P}$-completeness result, although achieving a similar result
for the depth of adaptive connectivity in the centralized setting for
batches of size $O(s) = O(n^{\epsilon})$ would be extremely surprising
since it would imply a polynomial-time algorithm for the
(Topologically Ordered) Circuit Value Problem with sub-linear depth
and therefore polynomial speedup.

\subsection{Organization}
%The rest of the paper is organized as follows:
%Definition of the \MPC{} model, as well as dynamic graph data
%structures, is in Section~\ref{sec:Prelims};
Section~\ref{sec:1conn} describes the full version of the high level idea for graph connectivity.
Section~\ref{sec:MPCDynamicTree} contains a discussion of the data structure we used to handle batch-update in constant round.
Section~\ref{sec:Contraction} gives a proof of our fast contraction lemma.
Section~\ref{sec:independentextractor} gives a proof of our independent sample extractor from sketches.
Section~\ref{sec:1conn_algorithms} presents the algorithm for graph connectivity and the correctness proof.
Lastly, we present our lower and upper bounds for the adaptive
connectivity problem in Section~\ref{sec:lowerbound}.

%\input{Prelims}
%\input{MPCDynamicTree}
%!TEX root = Main.tex

\section{1-Edge-Connectivity} \label{sec:1conn}

In this section we prove our result for 1-edge-connectivity, restated here:

\Main*

\paragraph{Parallel Batch-Dynamic Data Structure.}
Similar to previous results on dynamic connectivity~\cite{Frederickson85, GalilI92, HenzingerK99,HolmDT01,ahn2012analyzing,KapronKM13,
	GibbKKT15:arxiv,NanongkaiS17,Wulffnilsen17,NanongkaiSW17},
our data structure is based on maintaining a maximal spanning forest,
which we denote using $F$.
Formally, we define it as follows.
\begin{definition}[Maximal spanning forest]
	Given a graph $G$, we call $F$ a maximal spanning forest of $G$ if
	$F$ is a subgraph of $G$ consisting of a spanning tree in every
	connected component of $G$.
\end{definition}

Note that this is more specific than a spanning forest,
which is simply a spanning subgraph of $G$ containing no cycles.
This forest encodes the connectivity information in the graph,
and more importantly,
%undergoes at most one update per update
undergoes few changes per update to the graph.
Specifically:
\begin{enumerate}
	\item An \emph{insert} can cause at most two trees in $F$ to be joined to form
	a single tree.
	\item A \emph{delete} may split a tree into two, but if there
	exists another edge between these two resulting trees, they should
	then be connected together to ensure that the forest is maximal.
\end{enumerate}
Note that aside from identifying an edge between two trees formed
when deleting an edge from some tree, all other operations are tree
operations.
Specifically, in the static case, these operations can be entirely
encapsulated via tree data structures such as
dynamic trees~\cite{SleatorT83} or Top-Trees~\cite{AlstrupHLT05}.
We start by ensuring that such building blocks also exist in the
\MPC{} setting.
In Section~\ref{sec:MPCDynamicTree}, we show that a forest
can also be maintained efficiently in $O(1)$ rounds and
low communication in the \MPC{} model (Theorem~\ref{thm:Main_data_structure}).
In this section, we build upon this data structure and
show how to process updates and 1-edge-connectivity queries while maintaining
a maximal spanning forest of $G$.

Let $T(v)$ indicate the tree (component) in $F$ to which a vertex $v$ belongs.
We define the component ID of $v$ as the as the ID of this $T(v)$.
We represent the trees in the forest using the following data
structure. We describe the data structure in more detail in
Section~\ref{sec:MPCDynamicTree}.

\begin{restatable}[]{theorem}{MainDataStructure}
	%\begin{theorem}
	\label{thm:Main_data_structure}
	In the \MPC{} model with memory per machine $s = \Otil(n^\epsilon)$ for some constant $\epsilon$, for any constant $0 < \alpha < 1$ and a key length $\ell_{key}$ such that $n^\alpha \cdot \ell_{key} \leq s$,
	we can maintain a  dynamic forest $F$ in space $\Otil(n)$,
	with each vertex $v$ augmented with a key $\xx_v$ of length $\ell_{key}$($\xx_v$ is a summable element from a semi-group),
	%and two functions $f$ on edges and $g$ on pairs of vertices,
	%under the operations:
	\begin{itemize}
		\item $\textsc{Link}(u_1v_1, \ldots, u_kv_k)$: Insert a batch of $k$ edges into $F$.
		\item $\textsc{Cut}(u_1v_1, \ldots, u_kv_k)$: Delete $k$ edges from $F$.
		\item $\textsc{ID}(v_1, \ldots, v_k)$: Given a batch of $k$ vertices, return their component IDs in $F$.
		\item $\textsc{UpdateKey}((v_1, \xxhat_{1}'), \ldots, (v_k, \xxhat_{k}'))$: For each $i$, update the value of $\vec \xx_{v_i}$ to $\vec \xx_i'$.
		\item $\textsc{GetKey}(v_1, \ldots, v_k)$: For each $i$, return the value of $\vec \xx_{v_i}$.
		\item $\textsc{ComponentSum}(v_1 \ldots, v_k)$: Given a set of $k$ vertices, compute for each $v_i$,
		\[
		\sum_{w \colon w \in T(v_i)} \xx_{w}
		\]
		under the provided semi-group operation.
		%\item $\textsc{QueryFunction}((a_1, b_1), (a_2, b_2), \ldots, (a_k, b_k))$: Given $k$ pairs of vertices, output $g(a_i, b_i)$ for every $1 \leq i \leq k$.
		%\item $\textsc{UpdateFunction}((a_1, b_1, \hat{f}_1), (a_2, b_2, \hat{f}_2), \ldots, (a_k, b_k, \hat{f}_k))$: For each $i$, update the function value $f(e)$ using $\hat{f}_i$ for each tree edge $e$ on the path between $a_i$ and $b_i$.
	\end{itemize}

	Moreover, all operations can be performed in $O(1/\alpha)$ rounds and
	\begin{itemize}
	\item \textsc{Link} and \textsc{Cut} operations can be performed in $\Otil(k \cdot \ell_{key} \cdot n^\alpha)$ communication per round,
	\item \textsc{ID}  can be performed in $\Otil(k)$ communication per round,
	\item \textsc{UpdateKey}, \textsc{GetKey} and $\textsc{ComponentSum}$ operations can be performed in $\Otil(k \cdot \ell_{key} \cdot n^\alpha)$ communication per round.
	\end{itemize}

%	all operations can be performed in
%	$O(1 / \epsilon)$ rounds and {\color{red} $\Otil(k \cdot \ell_{key} \cdot  n^{\delta}  )$} communication per round.
	%The memory available to each machine is assumed to be $s = \Otil(n^{\epsilon})$, for some $0 < \epsilon < 1$.
	%\end{theorem}
\end{restatable}

Edge insertions and queries can be handled by above dynamic data structure:
for a set of edge queries, we use the \textsc{ID} operation to query the ID of all the vertices.
Two vertices are in the same connected component if and only if their IDs are same.
For a batch of edge insertions, we maintain the spanning forest by first identifying
all the inserted edges that join different connected components using \textsc{ID} operation, and then
using the \textsc{Link} operations to put these edges into the forest.

The process of handling a set of edge deletions is more complex.
This is because, if some spanning forest edges are removed,
then we need to find replacement edges in the graph which were previously not in the spanning forest,
but can be added to maintain the desired spanning forest.
To do this,
we use the the augmentation of tree nodes with $\xx_u$ and the
$\textsc{ComponentSum}$ operation  to accommodate each vertex storing
``sketches" in order to find replacement edges upon deletions.

%We will defer its discussion until we introduce our deletion operation.
%Note that as we will handle the presence of cycles before calling
%\textsc{Link}, this data structure does not need
%to address cycles in the input batch of edges.
%In addition to this dynamic forest $F$,
%we will also maintain a data structure to store all edges $E$ in $G$,
%along with their IDs and endpoints.
%
%
%In the incremental setting of only needing to handle
%edge insertions from \textsc{Insert},
%we can directly invoke the \textsc{Link} operation to maintain
%the spanning forest.
%However, we also need to search all edges when looking for a replacement
%for a deleted tree edge.
%So the insertion operation needs to deal with updating the sketch values
%needed for the same.
%To make this clear, we first discuss our deletion operation (\textsc{Delete}).

\paragraph{Sketching Based Approach Overview.}
At the core of the \textsc{Delete} operation is an
adaptation of the sketching based approach for finding replacement
edges by Ahn et al.~\cite{ahn2012analyzing} and Kapron et al.~\cite{KapronKM13}.
Since we rely on these sketches heavily,
we go into some detail about the approach here.
Without loss of generality, we assume every edge has a unique $O(\log n)$-bit ID, which is generated by a random function on the two vertices involved.
%While we utilize their scheme in a black-box manner, it is nonetheless
%useful to discuss some high level intuition about it.

For a vertex $v$, this scheme sets $\xx_v$ to the XOR
of the edge IDs of all the edges incident to $v$ (which we assume to be integers):
\[
\xx_{v} \defeq \xor_{\edge \colon \edge \sim v} \edge.
\]
For a subset of vertices $S$, we define $\partial(S)$ as the set
of edges with exactly one endpoint in $S$.
Then, taking the total XOR over all the
vertices in $S$ gives (by associativity of XOR)
\[
\xor_{v \in S} \xx_{v}
= \xor_{v \in S} \xor_{\edge \colon \edge \sim v} e
= \xor_{\edge \in E}
\left( \xor_{v \colon v \in S, \edge \sim v}  \edge \right)
=\xor_{\edge \in \partial(S)} \edge.
\]
So if there is only one edge leaving $S$, this XOR
over all vertices in $S$ returns precisely the ID of this edge.
To address the case with multiple edges crossing a cut,
Ahn \emph{et al.}~\cite{ahn2012analyzing} and Kapron \emph{et al.}~\cite{KapronKM13} sampled multiple subsets of edges
at different rates to ensure that no matter how many
edges are actually crossing, with high probability one sample
picks only one of them.
This redundancy does not cause issues because the edge query procedures also serve as a way to remove false positives. %The variant of this sketching result that we will use is stated as follows in Lemma 3.3.

We formally define the sketch as follows:
\begin{definition}[Graph Sketch from~\cite{ahn2012analyzing,KapronKM13}]\label{def:sketch}
A sketch with parameter $p$ of a graph $G = (V, E)$ is defined as follows:
\begin{enumerate}
\item Every edge is sampled independently with probability $p$.
Let $E'$ be the set of sampled edges.
\item For every vertex $v \in V$, let
\[
\xx_{v} \defeq \xor_{\edge \in E' \colon \edge \sim v} \edge.
\]
\end{enumerate}
\end{definition}
We say a sketch generates edge $\edge$ if there exists a vertex $v$ such that $\xx_v = \edge$.
%This redundancy does not cause issues because the query procedures also serve
%as a way to remove false positives, and we will do the same
%using the \textsc{1ConnQuery} procedure.
%\saurabh{I'm a bit unsure of what redundancy is meant here}
%\richard{I'm happy with this going into appendix, just figured it's worthwhile
%to try to explain from scratch?}
The variant of this sketching result that we will use is
stated as follows in Lemma~\ref{lem:KKM}.
\begin{lemma}[Graph Sketch from~\cite{ahn2012analyzing,KapronKM13}]
	\label{lem:KKM}
	Assume we maintain a sketch for each of $p \in \{1, 1/2, 1/4, \dots, 1/2^{\lceil 2 \ln \rceil - 1}\}$, and let $\vec \xx_v$ denote the sketches on vertex $v$,
%	There is a randomized procedure that for any graph $G$,
%	associates each node $v$ with $O(\log^{2}n)$ bits, $\xx_v$ so that
	\begin{itemize}
		\item upon insertion/deletion of an edge, we can maintain all $\vec \xx_v$'s
		in $O(\log^{2}n)$ update time;
		\item for any subset of vertices $S$, from the value
		\[
		\xor_{v \in S}\vec \xx_v,
		\]
		we can compute $O(\log{n})$ edge IDs so that for any edge
		$e \in \partial(S)$,
		the probability that one of these IDs is $e$ is at least
		${1}/{|\partial(S)|}$.
	\end{itemize}
\end{lemma}

%\noindent \textbf{Fast contraction lemma.}
\paragraph{Fast Contraction Lemma.}
As XOR is a semi-group operation, we can use these sketches in conjunction
with the dynamic forest data structure given in Theorem~\ref{thm:Main_data_structure}
to check whether a tree resulting from an edge deletion
has any outgoing edges.
In particular, $O(\log{n})$ copies of this sketch structure
allow us to find a replacement edge with high probability after
deleting a single edge in $O(1/\epsilon)$ rounds and $O(n^{\epsilon})$
total communication.
Our algorithm then essentially ``contracts" these edges found,
thus essentially reconnecting temporarily disconnected trees in $F$.

However, a straightforward generalization of the above method to deleting a batch
of $k$ edges results in an overhead of $\Theta(\log{k})$,
%This is because once we receive IDs of the edge candidates for replacement,
%we need to verify whether they span different trees using queries
%(i.e. the $\textsc{ID}$ operation) before updating
%our connectivity information.
%\saurabh{I don't fully understand the above sentence}
%
%Furthermore, even if we are able to avoid this check of false
%positives by designing better sketches,
because it's possible that this
random contraction process may take up to $\Theta(\log{k})$ rounds.
Consider for example a length $k$ path: if we pick $O(1)$ random
edges from each vertex, then each edge on the path is omitted by
both of its endpoints with constant probability.
So in the case of a path, we only reduce the number of remaining edges
by a constant factor in expectation, leading to a total of about
$\Theta(\log{k})$ rounds.
With our assumption of $s = O(n^{\epsilon})$ and queries arriving
in batches of $k \leq s$, this will lead to a round count that's
up to $\Theta(\log{n})$.

We address this with a natural modification motivated by the path
example: instead of keeping $O(\log{n})$ independent copies of the sketching
data structures, we keep $\widetilde O(n^{\delta})$ copies, for some small constant $\delta$,
which enables us to
sample $n^{\delta}$ random edges leaving each connected component
at any point.
As this process only deals with edges leaving connected components,
we can also view these connected components as individual vertices.
The overall algorithm then becomes a repeated contraction process
on a multi-graph: at each round,
each vertex picks $n^{\delta}$ random edges
incident to it, and contracts the graph along all picked edges.
Our key structural result %, which we prove in Section~\ref{sec:Contraction},
is a lemma
that shows that this process terminates
in $O(1 / \delta)$ rounds with high probability.
To formally state the lemma, we first define a random process of sampling edges in a graph.

\ContractionSamplingDef*

Below is our structural lemma, which we prove in Section~\ref{sec:Contraction}.

\Contraction*

%In this section, we assume we want to sample edges leaving vertices.
%The argument naturally generalizes to the case of sampling edges leaving disjoint subsets of vertices.

%\noindent \textbf{Independent sample extractor from sketches.}
\paragraph{Independent Sample Extractor From Sketches.}
%To facilitate finding replacement edges after a batch of deletions,
%we require an efficient algorithm to discover edges going between
%currently existing components.
%Viewing each component as a vertex in a multigraph,
%this amounts to simply resolving a multigraph into its connected components.
%We visualize this by ``contracting" the discovered edges
%until we are left with only singleton vertices - each representing a component
%in the original multigraph.
On a high level, our idea is to use $O(k \cdot \text{polylog}(n))$ independent sketches
to simulate the required \ContractionSampling process, and then apply Lemma~\ref{lem:contraction_main}.
However, we cannot do this naively, because \ContractionSampling requires the edges being sampled independently, 
whereas the sketch from Lemma~\ref{lem:KKM} does not satisfy this property.
Recall that the sketch generated at a vertex $v$ can correspond to an edge (say $uv$)
if no other edge adjacent to $v$ was sampled in the same sketch.
Consider an example where two edges $uv$ and $uw$ are sampled by the graph.
This means that no other edge from $v$ or $w$ can be sampled in that same sketch,
implying the sampling process is not independent.
%Here is an illustrative example: Consider graph $G = (V, E)$
%with $V = \{x_1, x_2, x_3, x_4, x_5, x_6\}$ and 
%$E=\{\edge_1= (x_1, x_2),\edge_2 = (x_1, x_3), \edge_3 = (x_1, x_4), \edge_4 = (x_2, x_5), \edge_5 = (x_2, x_6)\}$. 
%Assume a sketch sample all the edges independently with probability $1/2$.
%Let $S_{x_i}$ to be the random variable of the sampled edge for vertex $x_i$ if only one sampled edge incidents to $x_i$, otherwise $S_{x_i} = \bot$.
%We have $\Pr[S_{x_1} = e_2] = 1/8, \Pr[S_{x_2} = e_4] = 1/8$, but $\Pr[S_{x_1} = e_2 \text{ and } S_{x_2} = e_4] = 1/32$.

We would like to remark that this is not an issue for previous sketching based connectivity algorithms (e.g. \cite{ahn2012analyzing,KapronKM13}),  because in~\cite{ahn2012analyzing,KapronKM13},
each time, any current connected component only needs to find an arbitrary edge leaving the connected component.
In this way, if most current connected components find an arbitrary edge leaving the component, then after contracting connected components using sampled edges,  the total number of connected components reduce by at least a constant factor.
In this way, after $O(\log n)$ iterations, each connected component shrinks into a single vertex.
But in our case the contraction lemma requires edges being sampled independently.
Hence, we cannot directly apply Lemma~\ref{lem:contraction_main} on sketches.

%However, even though $O(\log n)$ sketches with different sampling rates allow to find incident edges for a constant fraction of all the vertices with high probability,
%the edges sampled for two adjacent vertices using a single sketch are not independent. Consider the following illustrative graph $G = (V, E)$
%with $V = \{x_1, x_2, x_3, x_4, x_5, x_6\}$ and
%$E=\{\edge_1= (x_1, x_2),\edge_2 = (x_1, x_3), \edge_3 = (x_1, x_4), \edge_4 = (x_2, x_5), \edge_5 = (x_2, x_6)\}$.
%Assume a sketch sample all the edges independently with probability $1/2$.
%Let $S_{x_i}$ to be the random variable of the sampled edge for vertex $x_i$ if only one sampled edge incidents to $x_i$, otherwise $S_{x_i} = \bot$.
%We have $\Pr[S_{x_1} = e_2] = 1/8, \Pr[S_{x_2} = e_4] = 1/8$, but $\Pr[S_{x_1} = e_2 \text{ and } S_{x_2} = e_4] = 1/32$.

To get around this issue, we construct an independent edge sample extractor from the sketches and show that with high probability,
this extractor will extract a set of independent edge samples that are equivalent to being sampled from a \ContractionSampling random process, as required by Lemma~\ref{lem:contraction_main}.
One key observation is that if the graph is bipartite, then sketch values on the vertices from one side of the bipartite graph are independent, because every edge sample is only related to one sketch value.
The high level idea of our extractor is then to extract bipartite graphs from sketches,
such that each edge appears in many bipartite graphs with high probability.
For each sketch, consider the following random process:
\begin{enumerate}
\item For each vertex of the graph, randomly assign a color of red or yellow. Then we can construct a bipartite graph with red vertices on one side, yellow vertices on the other side, and an edge is in the bipartite graph if and only if the color of one endpoint is red, and the other endpoint is yellow. Note that this step is not related to the process of sketch construction.
\item Independently sample every edge not in the bipartite graph with probability same as the probability of sampling used in the sketch.
\item For each red vertex whose incident edges were not sampled in Step 2,
independently sample every edge incident to the vertex in the bipartite graph with probability same as that used in the sketch.
\item Choose all the edges sampled in Step 3 which do not share a red vertex with any other sampled edge.
\end{enumerate}

We show that the edges obtained in Step 4 are sampled independently (conditioned on the outcome of Step 2).
Another way to see this independence is to partition all the independent random variables in the process of generating all the sketches
into two random processes $R_1$ and $R_2$ (based on the bipartite graph generated for each sketch) such that $R_1$ and $R_2$ are independent and simulate a required \ContractionSampling process in the following sense:

\begin{enumerate}
\item After implementing the random process $R_1$ and based on the outcome of $R_1$, define a \ContractionSampling process as required by Lemma~\ref{lem:contraction_main}.
\item The random process $R_2$ simulates the defined \ContractionSampling process in the following sense:
there is a partition of  the independent random variables of random process $R_2$ into groups satisfying the following conditions:
\begin{enumerate}
\item There is a bijection between  groups and random variables of the \ContractionSampling process.
\item For each group, there exists a function of the random variables in the group such that the function is equivalent to the corresponding random variable of the \ContractionSampling process.
\end{enumerate}
\end{enumerate}
Furthermore, all the edges sampled by the defined \ContractionSampling process are generated by the sketches (meaning that there exist a vertex and a sketch such that sketch on the vertex is the ID of the sampled edge).
In this way, we argue that the edges generated by all the sketches contains a set of edges generated by a \ContractionSampling process
so that we can apply Lemma~\ref{lem:contraction_main}.

%{Janardhan: I do not understand this statement:}
%We show that after first phase, with good probability, we can define a ContractionSampling process based on the outcome of the first phase
%such that the second phase implement the ContractionSampling process and 
%for every edge sampled by the ContractionSampling process, there exists a vertex and a sketch such that the sketch on the vertex is the ID of the sampled edge. 

More formally, we define the simulation between two random processes as follows.

\begin{definition}
We say 
a set of independent random variables $E_1, E_2, \dots, E_t$  simulates another set of independent random variables $F_1, F_2, \dots, F_\ell$ if 
there exists a set of random variables $U \subseteq \{E_1, E_2, \dots, E_t\}$ such that with constant probability,  after fixing all the random variables of $U$, 
there are $\ell$ subsets $T_1, T_2, \dots T_\ell \subseteq \{E_1, E_2, \dots, E_t\} \setminus U$ (depending on the outcome of the random process for $U$) satisfying
\begin{enumerate}
\item $T_1, \dots, T_\ell$ are mutually disjoint. 
\item For every $i \in [\ell]$, there exist a random variable which is a function of random variables in $T_i$, denoted as $f_i(T_i)$, such that $f(T_i)$ is same to the random variable $F_i$.
\end{enumerate}
\end{definition}
And we show that the process of generating $ O(k \log^3 n)$ sketches simulates the random process in the contraction lemma. 

%
%\begin{lemma}
%For a given bipartite graph, consider the process of sampling every edge with a fixed probability. 
%For any vertex $v$ on the left side of the bipartite graph, define random variable $E_v$ as 
%\[
%E_v = \begin{cases}
%	 \text{the sampled edge incident to } v   & \text{if only one edge incident to $v$ was sampled}\\
%	\emptyset & \text{otherwise}
%\end{cases}
%\]
%Then, all the random variables are independent. 
%\end{lemma}

\IndependenceExtractor*

\section{Batch-Dynamic Trees in \MPC{}}\label{sec:MPCDynamicTree}

\providecommand{\set}[1]{\left\{#1\right\}}

In this section we describe a simple batch-dynamic tree data structure
in the \MPC{} setting. Our data structure is based on a recently
developed parallel batch-dynamic data structure in the shared-memory
setting~\cite{tseng2018batch}. Specifically, Tseng et al. give a
parallel batch-dynamic tree that supports batches of $k$ links, cuts,
and queries for the representative of a vertex in $O(k\log(n/k + 1))$
expected work and $O(\log n)$ depth w.h.p. Their batch-dynamic trees
data structure represents each tree in the forest using an Euler-tour
Tree (ETT) structure~\cite{HenzingerK99}, in which each tree is
represented as the cyclic sequence of its Euler tour, broken at an
arbitrary point. The underlying sequence representation is a
concurrent skip list implementation that supports batch \emph{join}
and \emph{split} operations. Augmented trees are obtained by
augmenting the underlying sequence representation.

We show that the structure can be modified to achieve low
\emph{round-complexity} and \emph{communication} in the \MPC{} setting.
We now define the batch-dynamic trees interface and describe how to
extend the data structure into the \MPC{} setting.
The main difficulty encountered in the shared-memory setting is that
nodes are stored in separate memory locations and refer to each other
via pointers.  Therefore, when traversing the skip list at some level
$i$ to find a node's ancestor at level $i+1$, it requires traversing
all nodes that occur before (or after) it at level $i$. We show that
by changing the sampling probability to $1/n^{\epsilon}$, we can
ensure that each level has size $\tilde{O}(n^{\epsilon})$, each level
can be stored within a \emph{single machine} and thus this search can
be done within a \emph{single round}. The new sampling probability
also ensures that the number of levels is $O(1/\epsilon)$ w.h.p. which
is important for achieving our bounds.

\myparagraph{Batch-Dynamic Trees Interface}
A batch-parallel dynamic trees data structure represents a forest
$G=(V, E)$ as it undergoes batches of links, cuts, and connectivity
queries. A \emph{Link} links two trees in the forest. A \emph{Cut}
deletes an edge from the forest, breaking one tree into two trees.  A
\emph{ID} query returns a unique representative for the tree
containing a vertex. Formally the data structure supports the
following operations:
\begin{itemize}
  \item \textbf{$\textsc{Link}(\set{\set{u_1, v_1}, \ldots, \set{u_k,
    v_k}})$} takes an array of edges and adds them to the graph
    $G$. The input edges must not create a cycle in $G$.

  \item \textbf{$\textsc{Cut}(\set{\set{u_1, v_1}, \ldots, \set{u_k,
      v_k}})$} takes an array of edges and removes them from the
    graph $G$.

  \item \textbf{$\textsc{ID}(\set{u_1, \ldots,
    u_k})$} takes an array of vertex ids and returns an array
    containing the representative of each $u_i$. The
    \emph{representative} of a node, $r(u)$ is a unique value s.t.
    $r(u) = r(v)$ iff $u$ and $v$ are in the same tree.
\end{itemize}

Furthermore, the trees can be augmented with values ranging over a
domain $D$, and a commutative function $f : D^2 \rightarrow D$. The
trees can be made to support queries for the sum according to $f$ on
arbitrary subtrees, but for the purposes of this paper queries over
the entire tree suffice. The interface is extended with the following
two primitives:
\begin{itemize}
  \item \textbf{$\textsc{UpdateKey}(\set{\set{u_1, \hat{x}_1},
    \ldots, \set{u_k, \hat{x}_k}})$} takes an array of vertex id,
    value pairs and updates the value for $u_i$ to $\hat{x}_i$.

  \item \textbf{$\textsc{GetKey}(\set{u_1, \ldots, u_k})$} takes an
    array of vertex ids and returns an array containing the value of
    each $u_i$, $\hat{x}_i$.

  \item \textbf{$\textsc{ComponentSum}(\set{u_1, \ldots, u_k})$}
    takes an array of vertex ids and returns an array containing
    $\sum_{w : w \in T(u_i)} \hat{x_w}$ where $T(u_i)$ is the tree
    containing $u_i$, $\hat{x_w}$ is the value for node $w$, and the
    sum is computed according to $f$.
\end{itemize}

We show the following theorem in this section. Let $\delta$ be a
parameter controlling the size of the keys stored at each node and let
$\alpha$ be a parameter controlling the size of the blocks stored
internally within a single machine.

\begin{theorem}\label{thm:batch_aug_trees}
Let $\delta$ be a parameter controlling the keysize, and $\alpha$ be a
constant controlling the blocksize s.t. $\delta + \alpha < \epsilon$
and $0 < \alpha$.  Then, in the \MPC{} model with memory per machine $s =
\tilde{O}(n^{\epsilon})$ there is an augmented batch-dynamic tree data
structure in \MPC{} that supports batches of up to $k$ \textsc{Link},
\textsc{Cut}, \textsc{ID}, \textsc{UpdateKey},
\textsc{GetKey}, and \textsc{ComponentSum} operations in
$O(1/\alpha)$ rounds per operation w.h.p. where $k = O(n^{\alpha})$.

Furthermore, the batch operations cost
\begin{itemize}
  \item $\tilde{O}(kn^{\delta})$ communication per round w.h.p. for
    \textsc{UpdateKey}, \textsc{GetKey}, and \textsc{ComponentSum}
  \item $\tilde{O}(kn^{\delta}n^{\alpha})$ communication per round
    w.h.p. for \textsc{Link} and \textsc{Cut} and
  \item $O(k)$ communication per round for \textsc{ID}.
\end{itemize}
\end{theorem}

\subsection{Augmented Batch-Dynamic Sequences in \MPC{}}
In order to obtain Theorem~\ref{thm:batch_aug_trees}, we
first show how to implement augmented batch-dynamic sequences in
few rounds of \MPC{}. In particular, we will show the following lemma.
Note that achieving a similar bound on the round-complexity for large
batches, e.g., batches of size $O(n)$, would disprove the $2$-cycle
conjecture. We refer to~\cite{tseng2018batch} for the precise
definition of the sequence interface.

\begin{lemma}\label{lem:batch_aug_seq}
Let $\delta$ be a parameter controlling the keysize, and $\alpha$ be a
constant controlling the blocksize s.t. $\delta + \alpha < \epsilon$
and $0 < \alpha$. Then, in the \MPC{} model with memory per machine $s =
\tilde{O}(n^{\epsilon})$
there is an augmented batch-dynamic sequence
data structure in \MPC{} that supports batches of up to $k$
\textsc{Split}, \textsc{Join}, \textsc{ID},
\textsc{UpdateKey}, \textsc{GetKey}, and \textsc{SequenceSum}
operations in $O(1/\alpha)$ rounds per operation w.h.p. where $k =
O(n^{\alpha})$.

Furthermore, the batch operations cost
\begin{itemize}
  \item $\tilde{O}(kn^{\delta})$ communication per round w.h.p. for
    \textsc{UpdateKey}, \textsc{GetKey}, and \textsc{SequenceSum}
  \item $\tilde{O}(kn^{\delta}n^{\alpha})$ communication per round
    w.h.p. for \textsc{Split} and \textsc{Join} and
  \item $O(k)$ communication per round for \textsc{ID}.
\end{itemize}
\end{lemma}

For the sake of simplicity we discuss the case where $\delta = 0$ and
$0 < \alpha < \epsilon$ (i.e. values that fit within a constant number
of machine words), and describe how to generalize the idea to larger
values at the end of the sub-section.

\myparagraph{Sequence Data Structure} As in Tseng et
al.~\cite{tseng2018batch} we use a skip list as the underlying
sequence data structure. Instead of sampling nodes with constant
probability to join the next level, we sample them with probability
$1/n^{\alpha}$. It is easy to see that this ensures that the number
of levels in the list is $O(1/\alpha)$ w.h.p. since $\alpha$ is a
constant greater than $0$. Furthermore, the
largest number of nodes in some level $i$ that ``see" a node at level
$i+1$ as their left or right ancestor is $O(n^{\alpha}\log n)$
w.h.p. We say that the left (right) \emph{block} of a node belonging
to level $i$ are all of its siblings to the left (right) before the
next level $i+1$ node.
As previously discussed, in the \MPC{} setting we should intuitively try
to exploit the locality afforded by the \MPC{} model to store the blocks
(contiguous segments of a level) on a single machine.  Since each
block fits within a single machine w.h.p., operations within a block
can be done in 1 round, and since there are $O(1/\alpha)$ levels,
the total round complexity will be $O(1/\alpha)$ as desired. Since
the ideas and data structure are similar to Tseng et
al.~\cite{tseng2018batch}, we only provide the high-level details and
refer the reader to their paper for pseudocode.

\myparagraph{Join}
The join operation takes a batch of pairs of sequence elements to
join, where each pair contains the rightmost element of one sequence
and the leftmost element of another sequence. We process the levels
one by one. Consider a join of $(r_i, l_i)$. We scan the blocks for
$r_i$ and $l_i$ to find their left and right ancestors, and join them.
In the subsequent round, these ancestors take the place of $(r_i,
l_i)$ and we recursively continue until all levels are processed.
Observe that at each level, for each join we process we may create a
new block, with $\tilde{O}(n^{\alpha})$ elements. In summary, the
overall round-complexity of the operation is $O(1/\alpha)$ w.h.p.,
and the amount of communication needed is $\tilde{O}(kn^{\alpha})$
w.h.p.

\myparagraph{Split}
The split operation takes a batch of sequence elements at which to
split the sequences they belong to by deleting the edge to the right
of the element. We process the levels one by one.  Consider a split at
a node $e_i$. On each level, we first find the left and right
ancestors as in case of join. We then send all nodes splitting a given
block to the machine storing that block, and split it in a single
round. Then, we recurse on the next level. If the left and right
ancestors of $e_i$ were connected, we call split on the left right
ancestor at the next level. The overall round-complexity is
$O(1/\alpha)$ w.h.p., and the amount of communication needed is
$\tilde{O}(kn^{\alpha})$ w.h.p.

\myparagraph{Augmentation and Other Operations}
Each node in the skip list stores an augmented value which represents
the sum of all augmented values of elements in the block for which it
is a left ancestor. Note that these values are affected when
performing splits and joins above, but are easily updated within the
same round-complexity by computing the correct sum within any block
that was modified and updating its left ancestor. \textsc{SetKey} operations,
which take a batch of sequence elements and update the augmented
values at these nodes can be handled similarly in the same
round-complexity as join and split above. Note that this structure
supports efficient range queries over the augmented value, but for the
purposes of this paper, returning the augmented value for an entire
sequence (\textsc{SequenceSum}) is sufficient, and this can clearly be done in
$O(1/\alpha)$ rounds and $O(k)$ communication. Similarly, returning
a representative node (\textsc{ID}) for the sequence can be done in $O(1/\alpha)$
rounds w.h.p. and $O(k)$ communication by finding the top-most level
for the sequence containing the queried node, and returning the
lexicographically first element in this block.

\myparagraph{Handling Large Values}
Note that if the values have super-constant size, i.e. size
$O(n^{\delta})$ for some $\delta$ s.t. $\delta + \alpha < \epsilon$ we
can recover similar bounds as follows. Since the blocks have size
$\tilde{O}(n^{\alpha})$ and each value has size $O(n^{\delta})$ the
overall size of the block is $\tilde{O}(n^{\alpha + \delta}) =
\tilde{O}(n^\epsilon)$. Therefore blocks can still be stored within a
single machine without changing the sampling parameter. Storing large
values affects the bounds as follows. First, the communication cost of
performing splits and joins grows by a factor of $O(n^{\delta})$ due
to the increased block size. Second, the cost of getting, setting, and
performing a component sum grows by a factor of $O(n^{\delta})$ as
well, since $k$ values are returned, each of size $O(n^{\delta})$.
Therefore the communication cost of all operations other than finding
a represntative increase by a multiplicative $O(n^{\delta})$ factor.
Finally, note that the bounds on round-complexity are not affected,
since nodes are still sampled with probability $1/n^{\alpha}$.

\subsection{Augmented Batch-Dynamic Trees in \MPC{}}
We now show how to implement augmented batch-dynamic trees in \MPC{},
finishing the proof of Theorem~\ref{thm:batch_aug_trees}. We focus on
the case where $\delta = 0$ (we are storing constant size words) and
explain how the bounds are affected for larger $\delta$.

\myparagraph{Forest Data Structure}
We represent trees in the forest by storing the Euler tour of the tree in a
sequence data structure. If the forest is augmented under some domain
$D$ and commutative function $f : D^{2} \rightarrow D$, we apply this
augmentation to the underlying sequences.

\myparagraph{Link}
Given a batch of link operations (which are guaranteed to be acyclic)
we update the forest structure as follows. Consider a link $(u_i,
v_i)$. We first perform a batch split operation on the underlying
sequences at all $u_i, v_i$ for $1 \leq i \leq k$, which splits the
Euler tours of the underlying trees at the nodes incident to a link.
Next, we send all of the updates to a single machine to establish the
order in which joins incident to a single vertex are carried out.
Finally, we perform a batch join operation using the order found in
the previous round to link together multiple joins incident to a
single vertex. Since we perform a constant number of batch-sequence
operations with batches of size $O(k)$, the overall round complexity
is $O(1/\alpha)$ w.h.p.  by our bounds on sequences, and the overall
communication is $\tilde{O}(kn^{\alpha})$ w.h.p.

\myparagraph{Cut}
Given a batch of cut operations, we update the forest structure as
follows. Consider a cut $(u_i, v_i)$. The idea is to splice this edge
out of the Euler tour by splitting before and after $(u_i, v_i)$ and
$(v_i, u_i)$ in the tour. The tour is then repaired by joining the
neighbors of these nodes appropriately. In the case of batch cuts, we
perform a batch split for the step above. For batch cuts, notice that
many edges incident to a node could be deleted, and therefore we may
need to traverse a sequence of deleted edges before finding the next
neighbor to join. We handle this by sending all deleted edges and
their neighbors to a single machine, which determines which nodes
should be joined together to repair the tour. Finally, we repair the
tours by performing a batch join operation. Since we perform a
constant number of batch-sequence operations with batches of size
$O(k)$ the overall round complexity is $O(1/\alpha)$ w.h.p.  by our
bounds on sequences, and the overall communication is
$\tilde{O}(kn^{\alpha})$ w.h.p.

\myparagraph{Augmentation, Other Operations and Large Values}
Note that the underlying sequences handle updating the augmented
values, and that updating the augmented values at some nodes trivially
maps to an set call on the underlying sequences. Therefore the bounds
for \textsc{GetKey} and \textsc{SetKey} are identical to that of
sequences. Similarly, the bounds for \textsc{ID} are
identical to that of the sequence structure.
For super-constant size values, the bounds are affected exactly as in
the case for augmented sequences with large values. The communication
costs for all operations other than \textsc{ID} grow
by an $O(n^{\delta})$ factor and the round-complexity is unchanged.
This completes the proof of Theorem~\ref{thm:batch_aug_trees}.

%!TEX root = Main.tex
%\section{Fast contraction}
\section{Fast Contraction}
\label{sec:Contraction}
%\saurabh{Section name: Fast Contraction?}

The aim of this section is to prove Lemma~\ref{lem:contraction_main},
which is pivotal in proving the correctness of the main algorithm
from Section~\ref{sec:1conn}.
%\begin{lemma}
%	\label{lem:Contraction}
%	In a multigraph on $n$ vertices with at most $m$ (multi) edges, repeatedly
%	picking $n^{\epsilon}$ neighbors out of each vertex and contracting them
%	leads to the graph becoming singleton vertices (which may be disjoint)
%	in $O( 1 / \epsilon)$ iterations with high probability.
%\end{lemma}
%
%\Contraction*

Lemma~\ref{lem:contraction_main} is important in proving that our algorithm can find
replacement edges in the spanning forest quickly in the event
of a batch of edges being deleted.
The proof idea is as follows.
We first show that there exists a partitioning of the vertices
such that the edges within the partitions collapse in a single iteration.

To do this, we first need to define a few terms relating to
expansion criteria of a graph.
Let $\deg_G(v)$ denote the degree of a vertex $v$ in graph $G$.
For edges in a partition to collapse in a single iteration,
we need each partition to be sufficiently ``well-knit".
This property can be quantified using the notion of conductance.

%\begin{definition}[Conductance]
%Given a graph $G(V,E)$ and a subset of vertices $S \subseteq V$,
%the conductance of $S$ w.r.t. $G$ is defined as
%\[
%\phi_G(S) \defeq \min_{S' \subseteq S} \dfrac{\left| E(S', S \setminus S') \right|}{\min \left\{ \sum_{u \in S'} \deg_G(u), \sum_{u \in S \setminus S'} \deg_G(u) \right\}}.
%\]
%\end{definition}
\Conductance*

The following lemma proves the existence of a partitioning such that
each partition has high conductance.

%\begin{lemma}[\cite{SpielmanT11}, Section 7.1.]
%	\label{lem:Expander}
%	Given a parameter $k > 0$,  
%	any graph $G$ with $n$ vertices and $m$ edges can be partitioned into
%	groups of vertices $S_1, S_2, \ldots$ such that
%	\begin{itemize}
%		\item the conductance of each $S_i$ is at least $1/k$;
%		\item the number of edges between the $S_i$'s is at most $O(m \log{n}/k)$.
%	\end{itemize}
%\end{lemma}
\Expander*

Now that we have a suitable partitioning,
we want to find a strategy of picking edges in a decentralized fashion
such that all edges within a partition collapse with high probability.
One way to do this is to pick edges which form
a spectral sparsifier of $S_i$.
The following lemma by Spielman and Srivastava~\cite{SpielmanS08}
helps in this regard: we use more recent interpretations of it
that take sampling dependencies into account.

\begin{lemma}[\cite{SpielmanS08,Tropp12,KoutisLP16,KyngPPS17}]
	\label{lem:Sparsifier}
	On a graph $G$, let $E_1 \ldots E_{k}$ be independent random distributions
	over edges such that the total probability of an edge $e$ being picked
	is at least $\Omega(\log{n})$ times its effective resistance,
	then a random sample from $H = E_1 + E_2 + \ldots + E_k$ is connected
	with high probability.
\end{lemma}

Now we want to show that the random process ContractionSampling (Defintion~\ref{def:contractionsampling})
where each vertex draws $k \log^{2}n$ samples
actually satisfies the property mentioned in Lemma~\ref{lem:Sparsifier},
i.e., all edges are picked with probability at least
their effective resistance.
To show this, we first need the following inequality given by Cheeger.

\begin{lemma}[\cite{AlonM85}]
	\label{lem:Cheeger}
	Given a graph $G$, for any subset of vertices $S$ with conductance $\phi$,
	we have
	\[
	\lambda_2 \left( \DD_S^{-1/2} \LL_{S} \DD_S^{-1/2} \right) \geq \dfrac{1}{2} \phi^2,
	\]
	where $\LL_{S}$ is the Laplacian matrix of the subgraph of $G$ induced by $S$.
	$D_S$ is the diagonal matrix with degrees of vertices in $S$.
\end{lemma}

\begin{lemma}
	\label{lem:ERlessthanp}
	Let $S$ be a subset of vertices of $G$ such that $\phi_G(S) \geq 1/2\alpha^{1/3}$ for some $\alpha > 0$.
	For an edge $e=uv$, where $u,v \in S$, the effective resistance of $e$ measured in $S$, $ER_{S}(e)$, satisfies
	\[
	ER_{S}\left( e \right)
	\leq
	\alpha \left(\dfrac{1}{\deg_G(u)} + \dfrac{1}{\deg_G(v)}\right).
	\]
\end{lemma}

\begin{proof}
	From Lemma~\ref{lem:Cheeger}, we get that
	\[
	\LL_{S} \succeq \dfrac{1}{2}\left(\phi_G(S)\right)^2 \PPi_{\perp \vec{1}_S} D_S \PPi_{\perp \vec{1}_S}.
	\]
	Using this, along with the definition $ER_S(u,v) \defeq \chi_{uv}^T \LL_{S}^{\dagger} \chi_{uv}$,
	gives us that
	\begin{align}
	\label{eq:CheegerToDegree}
	ER_S(u,v) \leq \dfrac{1}{2} \left(\phi_G(S)\right)^{-2} \left(\dfrac{1}{\deg_S(u)} + \dfrac{1}{\deg_S(v)}\right).
	\end{align}
	We have for any subset $S' \subseteq S$ that:
	\begin{align*}
	\dfrac{E(S', S \setminus S')}{ \min \left\{ \sum_{u \in S'} \deg_G(u), \sum_{u \in S \setminus S'} \deg_G(u) \right\} } \geq \phi_G(S).
	\end{align*}
	
	Furthermore, for every vertex $v \in S$, we get
	\begin{align*}
	\dfrac{\deg_S(v)}{\deg_G(v)} \geq \phi_G(S),
	\end{align*}
	which when substituted into Equation~\ref{eq:CheegerToDegree} gives
	\begin{align*}
	ER_S(u,v) \leq \dfrac{1}{2} \left(\phi_G(S)\right)^{-3} \left(\dfrac{1}{\deg_G(u)} + \dfrac{1}{\deg_G(v)}\right).
	\end{align*}
	Substituting for $\phi_G(S) \geq 1/2 \alpha^{1/3}$ completes the proof.
%	
%	\saurabh{Check constants}
\end{proof}

Now, we have enough ammunition to prove Lemma~\ref{lem:contraction_main}.

\begin{proof}[Proof of Lemma~\ref{lem:contraction_main}]
	From Lemma~\ref{lem:Expander}, we know that our graph can be partitioned into
	expanders with conductance at least $\Omega(k^{-1/3} \log^{1/3}n)$.
	Now, let $S$ be one such partition and let $e = uv$ be an edge contained in $S$.
	From the definition of the random process in Definition~\ref{def:contractionsampling},
	we know that for an edge $uv$,
	the probability that it is sampled by either $u$ or $v$ is at least
	\[
	k \log^{2}n \left(\dfrac{1}{\deg_G(u)} + \dfrac{1}{\deg_G(v)}\right)
	\geq ER_S(uv) \cdot \Omega(\log n),
	\]
	where the inequality follows from Lemma~\ref{lem:ERlessthanp}.
	Since each such edge $uv$ is chosen with probability greater than $\Omega(\log{n})$ times its
	effective resistance w.r.t. $S$,
	from Lemma~\ref{lem:Sparsifier}, we know that the edges chosen
	within $S$ are connected with high probability.
	
	Thus, we are left with the edges between the partitions, the number of which
	is bounded by $O(m \log^{4/3}{n} \cdot k^{-1/3})$ edges,
%	Finally, repeating this contraction routine $O(1/\delta)$ times
%	contracts all edges with high probability.
\end{proof}

%!TEX root = Main.tex
\section{Independent Sample Extractor From Sketches}%Proof of Lemma~\ref{lem:independence}}
\label{sec:independentextractor}
In this section, we prove Lemma~\ref{lem:independence},
which shows how we extract independent edge samples
from the sketches, which are inherently dependent. 
%To illustrate our idea, we need a few more definition.
%We start with the definition of induced bipartite multigraph.

We start with the definition of an induced bipartite multigraph.
Given a multigraph $G = (V, E)$ of $n$ vertices, we say $B = (R, Y,  E_{B})$ is an \textbf{induced bipartite multigraph} of $G$ if 
$V$ is partitioned into two disjoint vertex sets $R$ and $Y$ and an edge of $G$ belongs to $E_B$ if and only if the edge contains one vertex from $R$ and one vertex from $Y$.

For a fixed multigraph $G$ and an induced bipartite multigraph $B$ of $G$, we conceptually divide the process of generating a sketch with parameter $p$ into two phases: 
\begin{enumerate}
\item[]Phase 1. Independently sample each edge not in the bipartite graph with probability $p$.
\item[]Phase 2.  Independently sample each edge in the bipartite graph with probability $p$.
\end{enumerate}

\begin{lemma}
Given a multigraph $G = (V, E)$ and an induced bipartite multigraph $B = (R, Y, E_B)$ of $G$, 
with  probability at least $1 - \frac{1}{n}$, $O(k\log^2 n)$ independent sketches simulate the following random process:
Every vertex $v \in R$ is associated with $t_v$ independent variables $S_{v, 1}, S_{v, 2}, \dots S_{v, t_v}$ for some integer $t_v \geq k$ satisfying 
		\begin{enumerate}
		\item %For any $i \in t_v$,
		The outcome of each $S_{v, i}$ can be an edge incident to $v$ or $\bot$. 
		\item 
		For every edge $\edge$ incident to vertex $v$, 
		\[\sum_{i = 1}^{t_v} \Pr[S_{v, i} = \edge] \geq \frac{2k}{\deg_G(v)}.\] 
		\end{enumerate}
Furthermore, for every edge sampled by the above random process, 
there exists a sketch and a vertex such that the value of the sketch on the vertex equals the edge ID.
\end{lemma}

	\begin{proof}
Assume for \[p \in \left\{\frac{1}{2}, \frac{1}{4}, \dots, \frac{1}{2^{\lceil \log_2 n\rceil}+1}\right\},\] 
that there are $t = 8000 k \log n$ sketches corresponding to each $p$.
Let $p_i$ denote the parameter of $i$-th sketch.

Let $m$ denote the number of edges in $G$. We use 
$E_{i, \edge_1}, E_{i, \edge_2}, \dots, E_{i, \edge_m}$ to denote the random variables denoting edges being present in the $i$-th sketch. 
Hence, the random process of generating all the sketches corresponds to sampling random variables $\{E_{i, \edge_j}\}_{i \in [t], j \in [m]}$.

Let $U \subseteq \{E_{i, \edge_j}\}_{i \in [t], j \in [m]}$ be the set of random variables in Phase 1 of all the sketches. 
We define another random process based on the outcome of $U$ as follows: 
For $i$-th sketch and any vertex $v \in R$, if no edge incident to vertex $v$ was sampled in Phase 1 of the $i$-th sketch, then 
we define a new independent random variable $S_{v, i}$ such that 
\[\Pr[S_{v, i} = \edge] = p_i (1 - p_i)^{\deg_B(v) - 1}\] if $\edge$ is in graph $B$ and incident to vertex $v$,
and \[\Pr[S_{v, i} = \bot] = 1 - p_i(1 - p_i)^{\deg_B(v) - 1} \cdot \deg_B(v).\]
If at least one edge  incident to vertex $v$ was sampled in Phase 1 of the $i$-th sketch, then we do not define random variable $S_{v, i}$.

Now, for an arbitrary $v \in R$, let 
\[p_v \defeq \frac{1}{2^{\lceil \log_2 \deg_G(v) \rceil + 1}}.\]
For a single sketch with parameter $p_v$, 
the probability that no edge incident to $v$ was sampled in Phase 1 is
\[(1 - p_v)^{\deg_G(v) - \deg_B(v)} \geq (1 - p_v)^{\deg_G(v) } > 0.1.\]
Applying Chernoff bound, with probability $1 - \frac{1}{n^3}$, at least $80k$ random variables $S_{v, i}$ are defined such that
with probability $p_v (1 - p_v)^{\deg_B(v) - 1}$,  $S_{v, j}$ equals exactly equals edge $\edge$ for every $\edge$ in $B$ incident to $v$. 
Hence, for any edge $\edge$ incident to $v$ in graph $B$, we have 
\[\sum_{S_{v, i}: S_{v, i} \text{ is defined}} \Pr[S_{v, i}  = \edge]\geq 80 k \cdot p_v (1 - p_v)^{\deg_B(v) - 1} > 80 k \cdot p_v (1 - p_v)^{\deg_G(v)} \geq \frac{2k}{\deg_G(v)}.\]
 By union bound, with probability $1 - \frac{1}{n}$, all the defined random variables $S_{v, j}$'s form the required random process.

In the rest of this proof, we show that Phase 2 of each sketch simulates the generation of the defined random varibles $\{S_{v, i}\}_{v \in R, i \in [t]}$.
For every defined random variable $S_{v, i}$, 
we let \[T_{v, i} = \{E_{i, \edge} : \text{$\edge \in B$ and is incident to vertex $v$}\}\] denote the random variable for the $i$-th sketch which corresponds to edges incident to vertex $v$ in graph $B$.
It is easy to verify that $T_{v, i} \cap U = \emptyset$. 
Furthermore, all the $T_{v, i}$'s are mutually disjoint. 
We define a function 
\[
f_{v, i}(T_{v, i}) = \begin{cases}
	 \edge  & \text{if } \sum_{E_{i, \edge'} \in T_{v, i}} E_{i, \edge'} = 1 \text{ and } E_{i, \edge} = 1\\
	\bot & \text{if } \sum_{E_{i, \edge'} \in T_{v, i}} E_{i, \edge'} \neq  1.
\end{cases}
\]
Since all the random variables in $T_{v, i}$ are independent, we have 
\[\Pr[f_{v, i}(T_{v, i}) = \edge] = p_i (1 - p_i)^{\deg_B(v) - 1}\]
for any edge $\edge$ incident to $v$ in $B$, 
and \[\Pr[f_{v, i}(T_{v, i}) = \bot] = 1 -  p_i (1 - p_i)^{\deg_B(v) - 1} \cdot \deg_B(v).\]
Then the lemma follows.
\end{proof}

Using the above lemma, we can now prove Lemma~\ref{lem:independence}.

\begin{proof}[Proof of Lemma~\ref{lem:independence}]
We repeat the following process $10 \log n$ times:
\begin{enumerate}
\item Every vertex is independently assigned the color red with probability 1/2, or is assigned yellow otherwise.
\item  Let $R$ be the vertices with red color and $Y$ be all the vertices with yellow color. Construct the induced bipartite multigraph $B = (R, Y, E_B)$, where $E_B$ contains all the edges of $G$ with one red vertex and one yellow vertex. 
\end{enumerate}
By Chernoff bound and union bound,  with probability at least $1 - \frac{1}{n}$, for every edge $\edge$ and a vertex $v$ contained by the edge $\edge$, there is a sampled bipartite multigraph $B= (R, Y, E_B)$ such that 
$v \in R$ and $\edge \in E_B$.

Assuming every edge belongs to at least one sampled bipartite graph. 
For each sampled bipartite multigraph, we assign $O(k \log^2 n)$ sketches.
The lemma follows by  
applying Lemma~\ref{lem:independence} for every bipartite multigraph and its assigned sketches, 
\end{proof}

\section{Connectivity Algorithms and Correctness}\label{sec:1conn_algorithms}

%In this section, we prove Lemma~\ref{lem:contraction_main}.
%We maintain the invariant that $F$ is a maximal spanning tree, and the key of each vertex is $O(k \log^3 n)$ sketches on the vertex.

We give the algorithms for batch edge queries, batch edge insertions, and batch edge deletions and prove the correctness in Section~\ref{sec:alg_query}, Section~\ref{sec:alg_insertion} and Section~\ref{sec:alg_deletion} respectively.
Putting together Lemmas~\ref{lem:Query}, \ref{lem:Delete} and \ref{lem:Insert}
then gives the overall result as stated in Theorem~\ref{thm:Main}.

Throughout this section, we will use the batch-dynamic tree data structure discussed in Section~\ref{sec:MPCDynamicTree} to maintain 
\begin{enumerate}
\item a maximal spanning forest $F$ of the graph,
\item a key $\vec \xx_v$ for every vertex $v$, where $\vec \xx_v$ is a vector of $\Otil(n^\delta)$ sketch values on vertex $v$,
\item an edge list data structure which can be used to check if an edge is in the graph given an edge ID.
\end{enumerate}

\subsection{Algorithm for Batch Edge Queries}\label{sec:alg_query}

%Once we have access to these data structures,
%we can build efficient algorithms for the
%operations from Theorem~\ref{thm:Main}.
Since $F$ is a maximal spanning tree,
the query operations are directly provided by calling $\textsc{ID}$
on all involved vertices.
Pseudocode of this routine is in Algorithm~\ref{alg:queryAlgo}.

\begin{algorithm}
	\begin{algbox}
		$\textsc{Query}((u_1,v_1), (u_2,v_2), \ldots, (u_k,v_k))$

		\textbf{Input}: Pairs of vertices $(u_1,v_1), (u_2,v_2), \ldots, (u_k,v_k)$

		\textbf{Output}: For each $1 \leq i \leq k$, \textbf{yes} if $u_i$ and $v_i$
		are connected in $G$, and \textbf{no} otherwise.
		\begin{enumerate}
			\item Call $\textsc{ID}
			(u_1, v_1, u_2, v_2, \ldots, u_k, v_k)$.
			\item For each $i$, output \textbf{yes} if $u_i$ and $v_i$ have the same component ID,	and \textbf{no} otherwise.
		\end{enumerate}
	\end{algbox}
	\caption{Querying the connectivity between a batch of vertex pairs}
	\label{alg:queryAlgo}
\end{algorithm}

\begin{lemma}
	\label{lem:Query}
	The algorithm \textsc{Query} (Algorithm~\ref{alg:queryAlgo}) correctly answers connectivity queries
	and takes $O(1/\alpha)$ rounds, each with total communication
	at most $\Otil(k)$.
\end{lemma}

\begin{proof}
	The correctness and performance bounds follow from the
	fact that $F$ is a maximal spanning forest of $F$ and from
	Theorem~\ref{thm:Main_data_structure}.
\end{proof}

\subsection{Algorithm for Batch Edge Insertions}\label{sec:alg_insertion}

Given  a batch of $k$ edge insertions,
we want to identify a subset of edges from the batch that are going to add to $F$ 
to maintain the invariant that $F$ is a maximal spanning forest.
To do this, we use $\textsc{ID}$ operation to find IDs of all the involved vertices in the edge insertion batch.
Then we construct a graph $G_{local}$ which initially contains all the edges in the edge insertion batch, and then contracts vertices from same connected component of $F$ to a single vertex.
Since this graph contains $k$ edges, we can put this graph into a single machine,
and compute a spanning forest $F_{local}$ of $G_{local}$.
We maintain the maximal spanning forest $F$ by adding edges in $F_{local}$ to $F$. %are the , and then identify all the edges belongs to $F$ after merging the batch of edge insertion.
We also maintain the edge list data structure by adding inserted edges to the list,
and maintain the sketches for the involved vertices by  the \textsc{UpdateKey} operation.
Pseudocode of the batched insertion routine is in Algorithm~\ref{alg:insertAlgo}.

\begin{algorithm}[!htb]
	\begin{algbox}
		$\textsc{Insert}(u_1v_1, u_2v_2, \ldots, u_kv_k)$

		\textbf{Input}: new edges $\edge_1 = u_1v_1, \edge_2 = u_2v_2, \ldots, \edge_k = u_kv_k$.

		\begin{enumerate}
			\item Add all $k$ edges to the edge list data structure.
			\item Run  $\textsc{GetKey}(u_1, v_1, \ldots, u_k, v_k)$.
			\item For every sketch, sample every inserted edge with probability equal to the parameter of the sketch, and compute the updated key value for vertices $\vec \xx_{u_1}', \vec \xx_{v_1}', \ldots, \vec \xx_{u_k}', \vec \xx_{v_k}'$.
			\item Run  $\textsc{UpdateKey}((u_1, \vec \xx_{u_1}'), (v_1, \vec \xx_{v_1}'), \ldots, (u_k, \vec \xx_{u_k}'), (v_k, \vec \xx_{v_k}'))$.
			\item Run $\textsc{ID}(\{u_1, v_1, u_2, v_2 \ldots u_k, v_k\})$.
			\item Using these IDs as vertex labels, construct a graph $G_{local}$ among
			the inserted edges, on a local machine.
			\item Find a maximal spanning forest $F_{local}$ of $G_{local}$
			locally on this machine.
			\item Run $\textsc{Link}(E(F_{local}))$.
		\end{enumerate}
	\end{algbox}
	\caption{Pseudocode for maintaining the data structure upon
		a batch of insertions.}
	\label{alg:insertAlgo}
\end{algorithm}

\begin{lemma}
	\label{lem:Insert}
	The algorithm \textsc{Insert} in Algorithm~\ref{alg:insertAlgo} correctly maintains a maximal spanning forest of $G$
	and takes $O(1/\alpha)$ rounds, each with total communication
	at most $\Otil(kn^{\alpha + \delta})$.
\end{lemma}

\begin{proof}
	To show the correctness, notice that since we add only a forest on
	the components as a whole,
	there is never an edge added between two already connected components.
	Additionally, since the forest is spanning, we do not throw away any necessary edges.

	From Theorem~\ref{thm:Main_data_structure}, using \textsc{GetKey}, \textsc{UpdateKey}, \textsc{ID} and \textsc{Link}
	falls under the claimed bound for rounds and communication,
	whereas the rest of the steps are performed only locally.
\end{proof}

\subsection{Algorithm for Batch Edge Deletions}\label{sec:alg_deletion}

%With this tool, we can now state our procedure for repairing
%the spanning forest after a batch of $k$ deletions.
Pseudocode of the batched deletion routine is in Algorithm ~\ref{alg:deleteAlgo}.
%The increase in the number of copies of sketches results in an extra
%overhead of $n^{\epsilon}$ in total memory, as well as
%communication per round.
%However, note that such overheads are already present in the
%underlying tree data structures from Theorem~\ref{thm:Main_data_structure}.

\begin{algorithm}[!htb]
	\begin{algbox}
		$\textsc{Delete}(e_1, e_2, \ldots, e_k)$

		\textbf{Input}: edges $\edge_1 = u_1v_1, \edge_2 = u_2v_2, \ldots, \edge_k = u_kv_k$ that are currently
		present in the graph.

		\begin{enumerate}
			\item Update the global edge index structure.
			\label{algline:removefromE}
			\item Run  $\textsc{GetKey}(u_1, v_1, \ldots, u_k, v_k)$.
			\item For every sketch, compute the updated key value  $\vec \xx_{u_1}', \vec \xx_{v_1}', \ldots, \vec \xx_{u_k}', \vec \xx_{v_k}'$ for vertices $u_1, v_1, \ldots, u_k, v_k$ by removing the IDs of edges $\edge_1, \ldots, \edge_k$.
			\item Run  $\textsc{UpdateKey}((u_1, \vec \xx_{u_1}'), (v_1, \vec \xx_{v_1}'), \ldots, (u_k, \vec \xx_{u_k}'), (v_k, \vec \xx_{v_k}'))$.
			\item Run $\textsc{Cut}$ for all edges that are in the spanning forest.
			Let $u_1 \ldots u_t$ be representative vertices from the resulting trees
			\item \label{algline:queryendpoints} Run $\textsc{ComponentSum}(\{ u_1 \ldots u_t\})$ to extract
			the total XOR values from each of the trees.
			
			\item Repeat $O(1/\delta)$ rounds:
			\label{algline:mainloop}
			\begin{enumerate}
				\item From the XOR values of the current components,
				deduce a list of potential replacement edges, $E_{R}$
				\item Identify the subset of edges with endpoints between current
				components given by $\textsc{ID}(u_1) \ldots \textsc{ID}(u_t)$ using a call to $\textsc{Query}$.
				\label{algline:otherendpoint}
				\item Find $T_R$, a maximal spanning forest of the valid replacement edges,
				via local computation.
				\item $\textsc{Link}(E(T_R))$.
				\item Update $u_1 \ldots u_t$ and their XOR values,
				either using another batch of queries, or by a local computation.
				\label{algline:insertreplacement}
			\end{enumerate}
		\end{enumerate}
	\end{algbox}
	\caption{Pseudocode for maintaining the data structure upon
		a batch of deletions.}
	\label{alg:deleteAlgo}
\end{algorithm}

%\richard{I stopped here: we might want to update signature and etc
%based on whether we stick with this presentation / set of function
%signatures}
%
\begin{lemma}
	\label{lem:Delete}
	The algorithm \textsc{Delete} (Algorithm~\ref{alg:deleteAlgo}) correctly maintains a maximal spanning forest of $G$
	and takes $O(1/\delta\alpha)$ rounds, each with total communication	at most $\Otil(kn^{\alpha + \delta})$.
\end{lemma}

\begin{proof}

	Note that $F$ remains a maximal spanning forest if the deleted edges are from
	outside of $F$.
	So, we only need to deal with the complementary case.
	Consider some tree $T \in F$, from which we deleted $\hat{k} - 1$ edges.
	$T$ is now separated into $\hat{k}$ trees, $T_1, T_2, \ldots , T_{\hat{k}}$.
	We need to show that the algorithm eventually contracts all $T_i$
	using the edges stored in the sketches.
	For this, note that the guarantees of Lemma~\ref{lem:independence}
	imply that from the $\Otil(n^{\delta})$ copies of sketches,
	we can sample edges leaving a group of $T_{i}$s in
	ways that meet the requirements of Lemma~\ref{lem:contraction_main}.
	These trees will collapse into singleton
	vertices in $O(1/\delta)$ rounds with high probability by applying Lemma~\ref{lem:contraction_main} iteratively.
	Thus the result is correct.

	Steps~\ref{algline:removefromE}-\ref{algline:queryendpoints} only require $O(1/\alpha)$
	rounds of communication, from Theorem~\ref{thm:Main_data_structure}.
	Step~\ref{algline:mainloop} loops $O(1/\delta)$ times,
	and its bottleneck is step~\ref{algline:otherendpoint},
	the verification of the locations of the endpoints in the trees.
	Once again by the guarantees of Theorem~\ref{thm:Main_data_structure},
	this takes $O(1/\alpha)$ rounds for each iteration,
	and at most $\Otil(kn^{\delta + \alpha})$ communication per round.

	Lastly, we call \textsc{Link} on the edges in $E_R$
	across various iterations.
	Since at most $k$ edges are deleted from $F$,
	there can only be at most $k$ replacement edges,
	so the total communication caused by these is $\Otil(kn^{\alpha + \delta})$.
\end{proof}

%
%Putting together Lemmas~\ref{lem:Query}, \ref{lem:Delete} and \ref{lem:Insert}
%then gives the overall result as stated in Theorem~\ref{thm:Main}.
%\input{contraction}
%\input{static}
%\input{RootedTree}
%\input{2Conn}
%\input{3Conn}
%!TEX root = Main.tex

\section{Adaptive Connectivity} \label{sec:lowerbound}

The adaptive connectivity problem is a ``semi-online" version of the
dynamic connectivity, where we are given a sequence of query/update
pairs, and each update (an edge insertion or deletion) is only applied
if its corresponding query evaluates to true on the graph resulting
from all operations before this pair. We say that the problem is
semi-online because although the entire input is known in advance, the
algorithm must answer each query taking into account all operations
that occur before it.

In this section, we show that this natural problem is \Pcomplete{}
under \NCone{} reductions. In the context of \MPC{} algorithms, our
result implies that if there exists an $O(1)$ round low-memory \MPC{}
algorithm solving the problem, then every problem in \classP{} can be
solved in $O(1)$ rounds in the low-memory \MPC{} model.

On the positive side, in Subsection~\ref{subsec:adaptiveupperbound} we
give an upper-bound based on the batch-dynamic connectivity algorithm
from Section~\ref{sec:1conn}, which shows that the adaptive
connectivity problem can be solved in $O(1)$ rounds for batches with
size proportional to the space per machine.

We first give a formal definition of the adaptive connectivity
problem.

\begin{definition}[Adaptive Connectivity]
The input to the \emph{Adaptive Connectivity} problem is an input
graph $G$ on $n$ vertices, and a sequence of query and update pairs:
$[(q_1, u_1), \ldots, (q_m, u_m)]$. Each query, $q_i$, is of the form
$\mathsf{Connected}(u,v)$ or $\lnot \mathsf{Connected}(u,v)$, and each
update, $u_i$, is either an edge insertion
($\mathsf{Insert}(e=(u,v))$) or an edge deletion
($\mathsf{Delete}(e=(u,v))$). The problem is to run each $q_i, i \in
[1, m]$ on the graph $G_i$, and apply $u_i$ to $G_i$ to produce
$G_{i+1}$ if and only if $q_i = \mathsf{true}$. The output of the
problem is $q_m$.
\end{definition}

\subsection{A Lower Bound for Adaptive Connectivity}\label{subsec:adaptivelowerbound}

We now prove our lower-bound, showing that the adaptive connectivity
problem is \Pcomplete{}. The idea is that we can use the adaptivity in
the problem to encode a circuit evaluation problem, which are well
known to be hard for \classP{}. Our reduction will be from the
Circuit Value Problem, defined below:

\begin{definition}[Circuit Value Problem]
  The input to the \emph{Circuit Value Problem} is an encoding of a
  circuit $C$ consisting of binary-fanin $\wedge$ ($\mathsf{and}$) and
  $\vee$ ($\mathsf{or}$) gates, and unary-fanin $\lnot$ ($\mathsf{not}$) gates,
  defined over n boolean inputs $x_1, \ldots, x_n$ with truth
  assignments.  Additionally there is a single specified output gate,
  $y$.  The problem is to evaluate $C$ and emit the value of the
  output gate, $y$.
\end{definition}

Our reduction makes use of a topological ordering of the input
circuit. A topological ordering of a DAG (e.g., circuit) is a
numbering $\rho$ of its vertices so that for every directed edge
$(u,v)$, $\rho(u) < \rho(v)$. Although we can topologically order a
DAG in $\mathsf{NC}^{2}$, there is no known \NCone{} algorithm for the
problem, which would mean that our reduction would use a (stronger)
$\mathsf{NC}^{2}$ reduction. To bypass this issue, we use the fact
that the Topologically-Ordered Circuit Value Problem is still
\Pcomplete{}~\cite{Greenlaw1995}. Therefore, in what follows we assume
that the circuit value problem instance provided to the reduction is
topologically ordered.

\AdaptiveLowerBound*
\begin{proof}
The adaptive connectivity problem is clearly contained in $\mathsf{P}$
since a trivial $O(\mathsf{poly}(m))$ work algorithm can first run a
connectivity query using BFS or DFS on $G_i$ to check whether the
vertices are connected or not, and then apply the update $u_i$
depending on the result of the query.

For hardness we give a reduction from the Topologically-Ordered
Circuit Value Problem. We assume the circuit $C$, is equipped with the
ability to query for the $i$-th gate in the specified topological
order in $O(1)$ time. Let $n$ be the number of variables in the
circuit, and $k$ be the number of gates.

The reduction builds the initial graph $G$ on $n+k+1$ vertices, where
there are $n$ vertices corresponding to the variables, $k$ vertices
corresponding to the gates, where gate $g_i$ corresponds to a vertex
$v_i$,  and a single distinguished root vertex, $r$. In the initial
graph, each variable that is set to $\mathsf{true}$ is connected to
$r$.

The reduction constructs a query/update sequence inductively as
follows. Consider the $i$-th gate in the topological ordering of $C$.
\begin{itemize}
  \item If the gate is of the form $g_i = g_a \wedge g_b$, we append
the following query/update pair to the sequence:
\begin{equation*}
  (\mathsf{Connected}(v_a, v_b), \mathsf{Insert}(r, v_i))
\end{equation*}
That is, if the vertex corresponding to gate $g_a$ is connected to the
vertex corresponding to gate $g_b$ in $G$, then add an edge between
the root $r$ and the vertex corresponding to the $i$-th gate.

\item Similarly, if the gate is $g_i = g_a \vee g_b$, we append the
following query/update pairs to the sequence:
\begin{align*}
  (\mathsf{Connected}(r, v_a), \mathsf{Insert}(r, v_i)) \\
  (\mathsf{Connected}(r, v_b), \mathsf{Insert}(r, v_i))
\end{align*}

\item Finally, if the gate is of the form $g_i = \lnot g_a$, we
append the following query/update pair to the sequence:
\begin{equation*}
  (\lnot \mathsf{Connected}(r, v_a), \mathsf{Insert}(r, v_i))
\end{equation*}
\end{itemize}

In this way a simple proof by induction shows that after executing all
query/update pairs in the sequence, the connected component in $G$
containing the root $r$ contains all vertices (gates) that evaluate to
$\mathsf{true}$. By making the final query of the form
\begin{align*}
  (\mathsf{Connected}(r, y), \_) \\
\end{align*}
where $y$ is the desired output gate, the output of the adaptive
connectivity instance returns the answer to the to input circuit.

It is easy to see that we can construct the query/update sequence in
\NCone{} as we can access the $i$-th gate independently in parallel,
and each gate can be made to emit exactly two update/query pairs (for
$\wedge$ and $\lnot$ gates we can simply insert a second noop
query/update pair).
\end{proof}

We have the following corollary in the \MPC{} setting.

\MPCLowerBound*
\begin{proof}
  The proof follows by observing that each of the \NCone{} reductions
  starting with the reduction from an arbitrary polynomial-time Turing
  machine, to the Topologically-Ordered Circuit Value Problem can be
  carried out in $O(1)$ rounds of \MPC{}. Therefore, by applying these
  reductions, we can transform any problem in \classP{} to an adaptive
  connectivity instance in $O(1)$ rounds of \MPC{}.
\end{proof}

\begin{remark}
  We note that there are no known polynomial-time algorithms for the
  (Topologically-Ordered) Circuit Value Problem with depth
  $O(n^{1-\epsilon})$, i.e., achieving polynomial speedup, and that
  finding such an algorithm has been a longstanding open question in
  parallel complexity theory~\cite{VitterS86, Condon94, Reinhardt97}.
  A parallel algorithm for adaptive connectivity in the centralized
  setting achieving even slightly sub-linear depth, e.g.,
  $O(n^{\epsilon - c})$ depth to process adaptive batches of size
  $O(n^{\epsilon})$ for any constants $\epsilon, c > 0$ would imply by
  our reduction above an algorithm for the (Topologically-Ordered)
  Circuit Value Problem with depth $O(n^{1-c})$, and therefore give an
  upper-bound with polynomial speedup.
\end{remark}

\myparagraph{Hardness for Other Adaptive Problems}
Note that the reduction given above for adaptive connectivity
immediately extends to problems related to connectivity, such as
directed reachability, and shortest-path problems. For adaptive
directed connectivity, when we add an $(x,y)$ edge in the undirected
case, we repeat the query twice and add both the $x\rightarrow y$ and
$y\rightarrow x$ edges. For adaptive unweighted shortest paths, if the
queries are of the form $\mathsf{DistanceLessThan}(u, v, d)$ then we
reduce these queries to connectivity/reachability queries by setting
$d$ to an appropriately large value (in the reduction above, setting
$d$ to $2$ suffices).

\subsection{An Upper Bound for Adaptive Connectivity}\label{subsec:adaptiveupperbound}

We now show that the \emph{static} batch-parallel 1-Edge-Connectivity
algorithm given in Theorem~\ref{thm:Main} can be used to solve the
adaptive connectivity problem. The bounds on the largest batch sizes
that the algorithm handle are identical to those in
Theorem~\ref{thm:Main}.

Given an adaptive batch of size $k$, the idea is to first take all
deletion updates in the batch, and ``apply" them on $G$ using a
modified version of Theorem~\ref{thm:Main}. Instead of permanently
inserting the newly discovered replacement edges into $G$, we
temporarily insert them to find all replacement edges that exist if
all deletions in the adaptive batch actually occur. This can be done
by first deleting the edges in the adaptive batch using
Theorem~\ref{thm:Main}, finding all replacement edges, and then
undoing these operations to restore $G$. The algorithm them collects
the adaptive batch, and all replacement edges (which have size at most
equal to the size of the adaptive batch) on a single machine, and
simulates the sequential adaptive algorithm on the contracted graph
corresponding to vertices active in the batch in 1 round. The
insertions and deletions that ensue from processing the adaptive batch
can be applied in the same bounds as Theorem~\ref{thm:Main}.
Therefore, we have an algorithm for adaptive connectivity with the
following bounds:

\AdaptiveUpperBound*

%\begin{remark}\label{rmk:adaptiveupperbound}
%In the \MPC{} model with memory per machine $s = \Otil(n^\epsilon)$ for some constant $\epsilon$,
%we can maintain a dynamic undirected graph with
%at most $m$ edges
%which can handle the following operation with high probability: (for any constant $\delta < \epsilon$, and $k \leq n^{\epsilon - \delta}$)
%\begin{enumerate}
%  \item An adaptive batch of up to $k$ (query, edge
%    insertions/deletions) pairs, using $O(1/(\epsilon\delta))$ rounds.
%\end{enumerate}
%Furthermore, the total communication for handling a batch of $k$
%operations is $\Otil(k n^{\delta})$, and the total space used across
%all machines is $\Otil(m + n^{1+\delta})$.
%\end{remark}

%\input{MoreGraph}

\bibliographystyle{alpha}
\bibliography{Ref}

\end{document}